\documentclass[
final            
,numberedheadings 
]{aipproc}

\layoutstyle{6x9}

\usepackage{amssymb}
\usepackage{amsmath}
\usepackage{amsthm}
\usepackage{pb-diagram}

\newcommand{\Tr}{\mathrm{\text{Tr}}}
\newcommand{\RM}{\mathbb{R}}
\newcommand{\ZM}{\mathbb{Z}}

\newcommand{\NM}{\mathbb{N}}

\newcommand{\PM}{\mathbb{P}}
\newcommand{\SM}{\mathbb{S}}
\def\e{\mathrm{e}}
\def\i{\mathrm{i}}
\def\d{\mathrm{d}}
\newtheorem{theorem}{Theorem}
\newtheorem{lemma}{Lemma}

\begin{document}

\title{Classical and quantum aspects of tomography}


\author{Paolo Facchi}{
  address={Dipartimento di Matematica, Universit\`a di Bari,
        I-70125  Bari, Italy}
  ,altaddress={INFN, Sezione di Bari, I-70126 Bari, Italy}
}

\author{Marilena Ligab\`o}{
  address={Dipartimento di Matematica, Universit\`a di Bari,
        I-70125  Bari, Italy}
}

\begin{abstract}
We present here a set of lecture notes on tomography.
The Radon transform and some of its generalizations are considered and their inversion formulae are proved.
We will also look from a group-theoretc point of view at the more general problem of expressing a function on a manifold in terms of its integrals over certain submanifolds. Finally, the extension of the tomographic maps to the quantum case is considered, as a Weyl-Wigner quantization of the classical case.
\end{abstract}

\maketitle

\section{Introduction}

We present here the notes of three lectures given by one of us at the XVIII International Fall Workshop on Geometry and Physics in Benasque. The course considers some aspects of tomographic mappings, a vast and very active area of different applications of integral geometry.

As a motivation, we start with the physical problem of the tomographic imaging of a two dimensional body by the measurements of the intensity attenuation of a beam of particles  passing through it. Then, we will introduce the Radon transform, which is the key mathematical tool for reconstructing the tomographic map of both the Wigner quasi-distribution of a quantum state and the probability distribution on the phase space of a classical particle. The original transform was introduced by J.~Radon \cite{Rad1917}, who proved that a differentiable function on  the three-dimensional Euclidean space can be determined explicitly by means of its integral over the planes. We will consider the straightforward generalization of the Radon transform to codimension-one hyperplanes in $\RM^n$ and will study its main properties. Moreover, we will prove the original Radon inversion formula, by making use of potential theory.

We will then look at a first generalization of the Radon transform: the $M^2$-transform which, roughly speaking, is defined on a larger space. We will show how it can be inverted very easily by using the Fourier transform. Then, we will unveil the relation between Radon and $M^2$, thus exhibiting an alternative, much simpler inversion formula.

Then, we will consider a broader framework and look at the more general problem of  expressing a function on a manifold in terms of its integrals over certain submanifolds. This has become an important topic in integral geometry with  many applications ranging from partial differential equations, group representations, and X-ray technology. The focus will be on invariant (or equivariant) transformations under some symmetry groups from the space of functions on one geometrical space to the space of functions on another geometrical space. We will also exhibit two examples that can be framed in the above setting and work out their respective inversion formulae: the Radon transform on codimension-$d$ planes and  the John transform on unit spheres in $\RM^3$.

We will then show a simple mechanism that allows to generalize the Radon transform, and its inverse, to more general families of codimensione-one submanifolds. As examples, we will look at two families of curves in the plane: the family of all circles passing through a fixed point and at the family of all equilateral hyperbolas with one fixed asymptote.

Finally, we will show how the above techniques can be used in the quantum case. A straightforward generalization derives from the phase-space description of quantum mechanics through Wigner quasi-distribution functions. We will see how the homodyne detection scheme  can be considered as the Weyl quantization of the Radon transform.

There is a huge literature on tomography and it is impossible to be complete in the references. We apologize in advance for any shortcomings in this respect.
There are several classical references: we suggest for a group-theoretic approach the book by Helgason \cite{Helgason} and for  the connection  with partial differential equations the book by John \cite{John}. Additional general references are Gel'fand and Shilov \cite{Gelf} and Strichartz \cite{Strichartz}.
Finally, for an up-to-date overview of the many applications of the Radon transform we refer to the book by Deans \cite{Deans}, which includes also an English translation of the 1917 paper by Radon, that announced the transform's discovery.

\begin{figure}[t]
\includegraphics[width=0.5\textwidth]{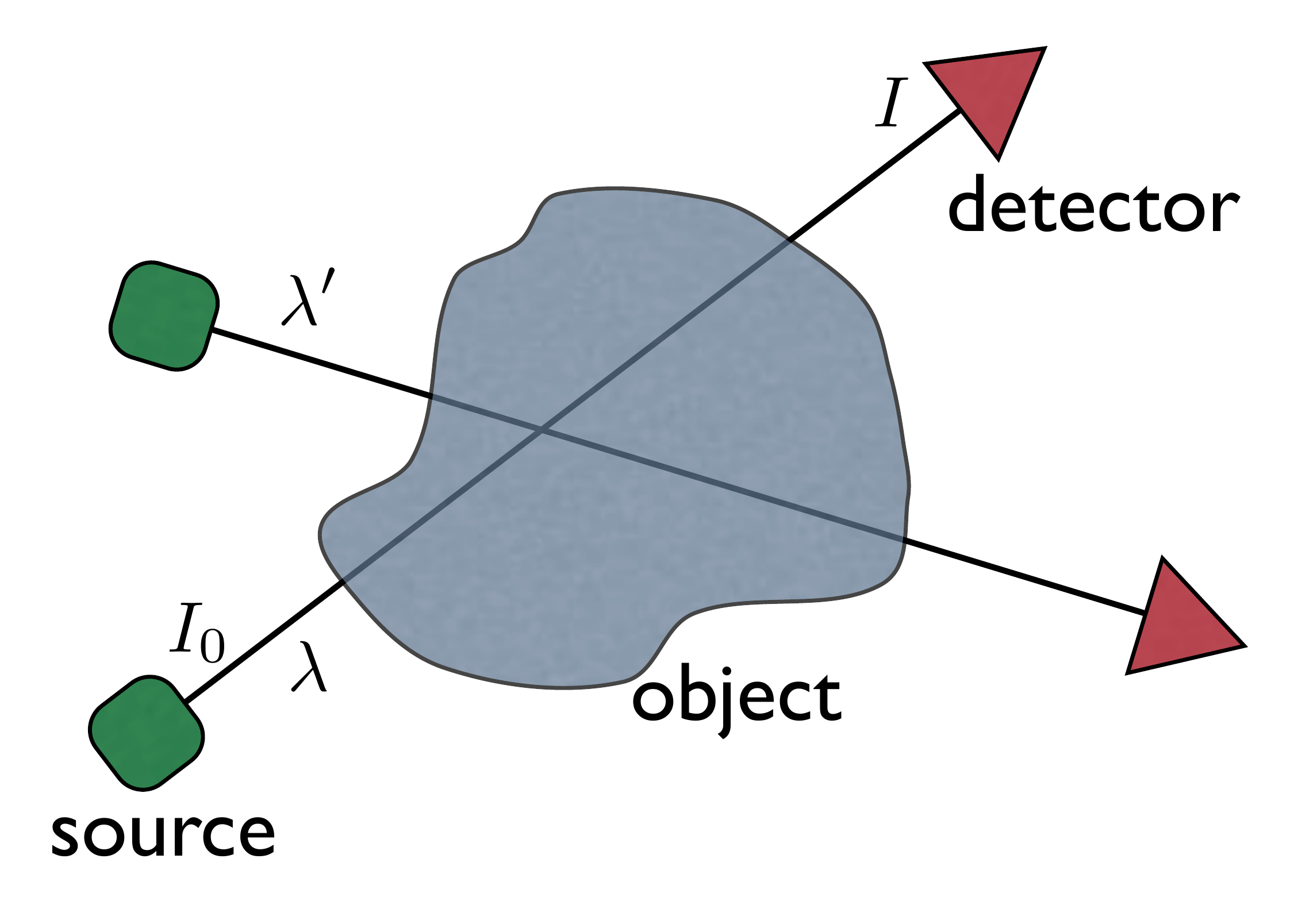}
\caption{Tomographic experimental setup.}
\label{detector}
\end{figure}

\section{Tomographic imaging in the plane}

Consider a body in the plane $\RM^2$, and consider a beam of particles (neutrons, electrons, X-rays, etc.) emitted by a source. Assume that the initial intensity of the beam is $I_{0}$. When the particles pass through the body they are absorbed or scattered and the intensity of the beam traversing a length $\Delta s$ decreases by an amount  proportional to the density of the body $\mu$, namely
\begin{equation}
\Delta I=-\Delta s\, \mu(s),
\end{equation}
so that
\begin{equation}
I(s)=I_{0}\exp \left(-\int_{0}^{s}\mu(r)\;\d r \right).
\end{equation}
A detector placed at the exit of the body measures the final intensity $I(s)$, and then from
\begin{equation}
-\ln \frac{I(s)}{I_0} =\int_{0}^{s} \mu(r)\; \d r
\end{equation}
one can record the value of the density integrated on a line.
If another ray with a different direction is considered, with the same procedure one obtains the value of the integral of the density on that line. See Fig \ref{detector}.

The mathematical model of the above setup is the following:
Given a normalized positive and smooth function $f=f(q,p)$, namely $f(q,p)\geq 0$ and
\begin{equation}
\int_{\RM^2} f(q,p)\; \d q\, \d p=1,
\end{equation}
define for all possible lines $\lambda$ in the plane $\RM^2$ the quantity
\begin{equation}\label{sharp lambda}
f^{\sharp}(\lambda)=\int_{\lambda}  f(x)\; \d m(x),
\end{equation}
where $\d m$ is the Euclidean measure on the line $\lambda$.
In this way, we have defined an operator $\sharp$ that maps a smooth function $f$ on the plane $\RM^2$ into a function $f^{\sharp}$ on $\PM^{2}$, the manifold of the lines in $\RM^2$. Along with the transform $f \mapsto f^\sharp$, it is convenient to consider also the dual transform $g \mapsto g^\flat$ which to a function $g$ on $\PM^{2}$ associates the function $g^\flat$ on $\RM^2$ given by
\begin{equation}
g^\flat(x)=\int_{x \in \lambda} g(\lambda) \; \d \mu(\lambda),
\end{equation}
where $\d \mu$ is the measure on the compact set $\{\lambda \in \PM^{2}|x \in \lambda\}$ which is invariant under the group of rotations around $x$ and for which the measure of the whole set is $1$.

\begin{figure}[t]
\includegraphics[width=0.5\textwidth]{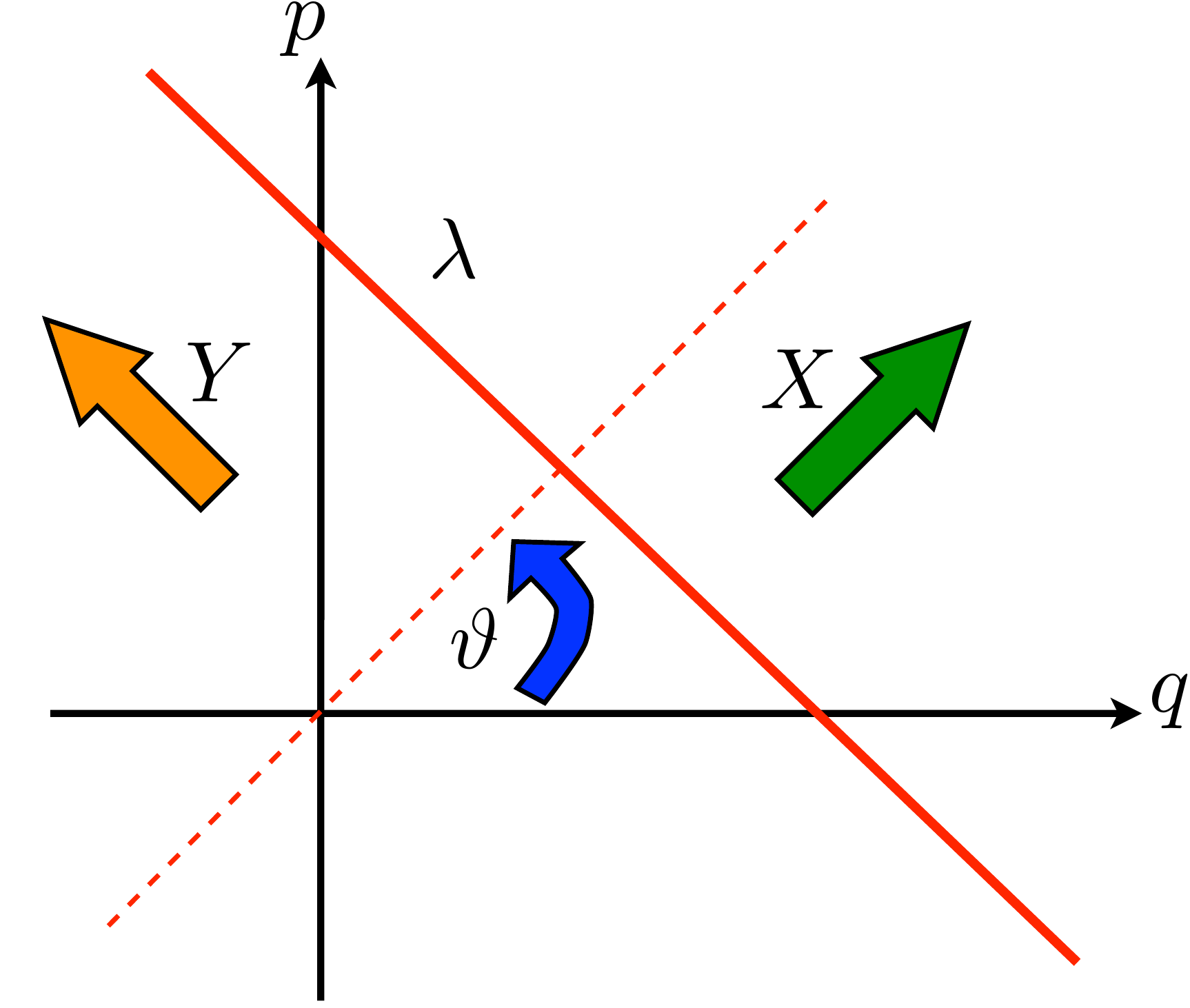}
\caption{Parametrization of $\lambda$ using its signed distance $X$ from the origin and a unit vector $\xi=(\cos \theta,\sin \theta)$ perpendicular to the line.}
\label{fig:radon 2d}
\end{figure}

We ask the following question: If we know the family of integrals $(f^{\sharp}(\lambda))_{\lambda \in \PM^{2}}$, can we reconstruct the total density function $f$? The answer is affirmative and in the following we will see how to obtain this result.

In order to write  the integral (\ref{sharp lambda}) in a more convenient form, let us parameterize $\PM^2$ in the following way
\begin{equation}
\RM \times \SM \to \PM^2, \qquad (X,\xi)\mapsto \lambda
\label{eq:doublecovering}
\end{equation}
where $X \in \RM$ is the (signed) distance of the line $\lambda$ from the origin, and $\xi=(\cos \theta, \sin \theta)\in\SM$ is a unit vector perpendicular to $\lambda$. This means that the Cartesian equation of $\lambda$ is
\begin{equation}
\lambda=\{(q,p)\in\RM^2 \,|\, X-q\cos \theta-p\sin \theta=0\}.
\end{equation}
By considering a pair $(X,Y)$ of Cartesian axes associated to the line $\lambda$, where $X$ is parallel to $\xi$, and $Y$ is perpendicular to $\xi$, we obtain the following parametric equation for $\lambda$
\begin{equation}
q=X\cos \theta-Y \sin \theta, \qquad p=X\sin \theta+Y\cos \theta.
\end{equation}
See Fig. \ref{fig:radon 2d}.
With a little abuse of notation, we have
\begin{equation}\label{radon transform 2d}
f^{\sharp}(\lambda)=f^{\sharp}(X,\xi)=\int_{\RM} \d Y \; f(X\cos \theta-Y \sin \theta,X\sin \theta+Y\cos \theta),
\end{equation}
that is called the \emph{Radon transform} of $f$.
Observe that the (\ref{eq:doublecovering}) is a double covering of $\PM^2$, because the pairs $(X,\xi)$ and $(-X,-\xi)$ correspond to the same line $\lambda$. However, one has
\begin{equation}
f^{\sharp}(X,\xi)=f^{\sharp}(-X,-\xi)
\end{equation}
and thus $f^{\sharp}$ depends only on the line $\lambda$ and not on its parameterization.

\section{Radon transform}\label{sec: Radon}

Let us generalize the above definitions to the case of a $n$-dimensional space.
Let $f$ be a function defined on $\RM^n$, integrable on each hyperplane in $\RM^n$ and let $\PM^n$ be the manifold of all hyperplanes in $\RM^n$. The Radon transform of $f$ is defined as follows
\begin{equation}\label{sharp lambda1}
f^{\sharp}(\lambda)=\int_{\lambda} f(x)\; \d m(x) ,
\end{equation}
where $\d m$ is the Euclidean measure on the hyperplane $\lambda$.
Thus we have an operator $\sharp$, the \emph{Radon transform}, that maps a function $f$ on $\RM^n$ into a function $f^{\sharp}$ on $\PM^{n}$, namely $f \to f^\sharp$. Its dual  transform, also called \emph{back projection operator}, $g \mapsto g^\flat$ associates to a function $g$ on $\PM^{n}$ the function $g^\flat$ on $\RM^n$ given by
\begin{equation}\label{dual map}
g^\flat(x)=\int_{x \in \lambda} g(\lambda) \; \d\mu(\lambda),
\end{equation}
where $\d\mu$ is the unique probability measure on the compact set $\{\lambda \in \PM^{n}|x \in \lambda\}$ which is invariant under the group of rotations around $x$.

\begin{figure}[t]
\includegraphics[width=0.45\textwidth]{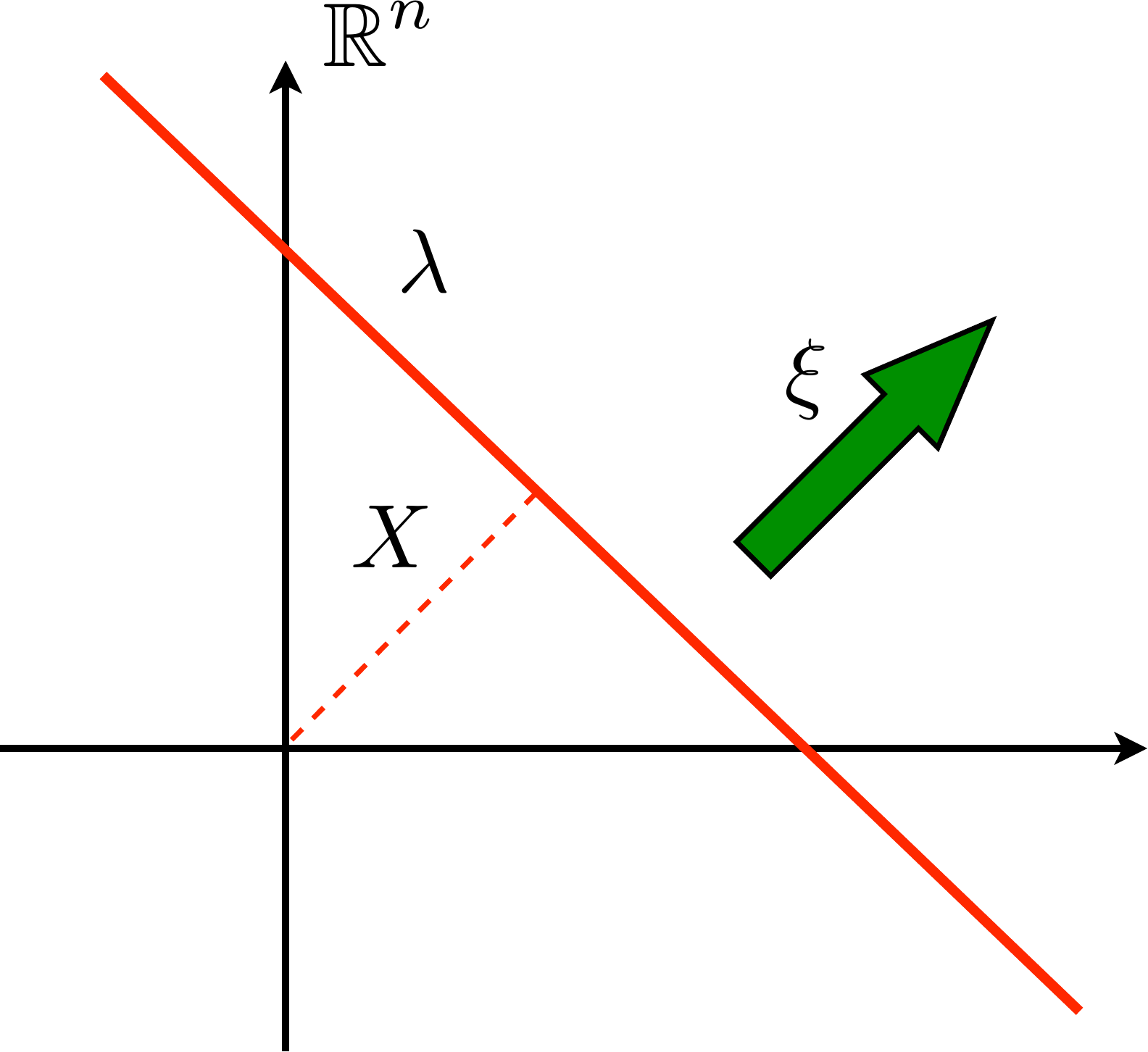}
\caption{Parametrization of the hyperplane $\lambda$ using its signed distance $X$ from the origin and a unit vector $\xi$ perpendicular to $\lambda$.}
\label{fig:radon nd}
\end{figure}

Let us consider the following covering of $\PM^n$
\begin{equation}\label{double covering}
\RM \times \SM^{n-1} \to \PM^n, \qquad (X,\xi)\mapsto \lambda,
\end{equation}
where $\SM^{n-1}$ is the unit sphere in $\RM^n$. Thus, the
equation of the hyperplane $\lambda$ is
\begin{equation}
\lambda=\{x\in\RM^n \, | \,  X-\xi \cdot x=0\},
\end{equation}
with $a\cdot b$ denoting the Euclidean inner product of $a, b\in\RM^n$. See Fig. \ref{fig:radon nd}.
Observe that the pairs $(X,\xi),(-X,-\xi) \in \RM \times \SM^{n-1}$ are mapped into the same hyperplane $\lambda \in \PM^n$. Therefore (\ref{double covering}) is a double covering of $\PM^n$.
Thus $\PM^n$ has a canonical manifold structure with respect to which this covering mapping is differentiable. We thus identify continuous (differentiable) functions $g$ on $\PM^n$ with continuous (differentiable) functions $g$ on $\RM \times \SM^{n-1}$ satisfying $g(X,\xi)=g(-X,-\xi)$.

We will work in the Schwartz space $\mathcal{S}(\RM^n)$ of complex-valued rapidly decreasing functions on $\RM^n$. We recall that $f \in \mathcal{S}(\RM^n)$ if and only if $f \in C^{\infty}(\RM^n)$ and for all multiindices  $\alpha=(\alpha_1,\dots,\alpha_n) \in \NM^n$ and $m \geq0$ results
\begin{equation}
\sup_{x \in \RM^n}\left| |x|^m D^\alpha f (x) \right| < +\infty,
\label{eq:Schwartz}
\end{equation}
where $|x|$ denotes the Euclidean $x$, and
\begin{equation}
D^\alpha f (x):=(\partial_{x_1}^{\alpha_1}\dots \partial_{x_n}^{\alpha_n}f)(x).
\end{equation}

In analogy with $\mathcal{S}(\RM^n)$ we define $\mathcal{S}(\RM \times \SM^{n-1})$ as the space of $C^{\infty}$ functions $g$ on $\RM \times \SM^{n-1}$ which for any integers $m,\alpha \geq 0$ and any differential operator $D$ on $\SM^{n-1}$ satisfy
\begin{equation}
\sup_{X \in \RM , \xi \in \SM^{n-1}}\left| |X|^m \frac{\partial^\alpha(Dg)}{\partial X^\alpha}(X,\xi)\right| < +\infty.
\end{equation}
The space $\mathcal{S}(\PM^n)$ is then defined as the set of $g \in \mathcal{S}(\RM \times \SM^{n-1})$ satisfying $g(-X,-\xi)=g(X,\xi)$.

If $f \in \mathcal{S}(\RM^n)$ its Radon transform (\ref{sharp lambda1}) can be written in a simple way, by using the Dirac $\delta$ distribution, as
\begin{equation}\label{Radon transform nd}
f^{\sharp}(X,\xi)=\int_{\RM^n} f(x)\, \delta(X-\xi\cdot x)\; \d x=\langle \delta(X-\xi\cdot x) \rangle_f,
\end{equation}
In the same way if $g \in \mathcal{S}(\PM^n)$ one has
\begin{equation}\label{flat}
g^{\flat}(x)=\int_{\RM\times \SM^{n-1}} g(X,\xi)\, \delta(X-\xi \cdot x)\; \d X \d\xi.
\end{equation}
Using (\ref{Radon transform nd}) and (\ref{flat}) it is easy to check the duality relation between the Radon transform $\sharp$ and its dual $\flat$, namely for every $f \in \mathcal{S}(\RM^n)$ and $g \in \mathcal{S}(\PM^n)$ one gets
\begin{equation}
\int_{\RM\times \SM^{n-1}} \d X \d\xi\; f^{\sharp}(X,\xi)\, g(X,\xi)= \int_{\RM^{n}} \d x \; f(x)\,  g^{\flat}(x).
\end{equation}
Indeed,
\begin{eqnarray}
\int_{\RM\times \SM^{n-1}} \d X \d\xi\; f^{\sharp}(X,\xi)\, g(X,\xi) & = & \int_{\RM\times \SM^{n-1}} \d X \d\xi  \int_{\RM^n} \d x\;  f(x)\, \delta(X-\xi\cdot x)\, g(X,\xi) \nonumber\\
& = & \int_{\RM^n}\d x\; f(x)\, g^{\flat}(x).
\end{eqnarray}

Now, consider a unit vector $\xi \in \SM^{n-1}$. We define \emph{the tomogram of $f$ along $\xi$} the function
\begin{equation}\label{tomogram}
f^{\sharp}_{\xi}(X):=f^{\sharp}(X,\xi).
\end{equation}
Physically, this means that we fix a direction $\xi$ and we look at the variation of the intensity of all the rays that interact with the body in that direction.
We also define the family $(f^{\sharp}_{\xi})_{\xi \in \SM^{n-1}}$ of tomograms associated to $f$.

The following result shows how the Radon transform maps a joint probability density function $f$ of $n$ random variables into a family of probability density functions $(f^{\sharp}_{\xi})_{\xi \in \SM^{n-1}}$ of a random variable.
\begin{lemma}
Let $f \in \mathcal{S}(\RM^n)$ be a nonnegative and normalized function, namely
\begin{equation}
f \geq 0 \qquad \mathrm{and} \qquad \int_{\RM^n}f(x) \; \d x=1.
\end{equation}
Then, for all $\xi \in \SM^{n-1}$ it results that
\begin{equation}
f^\sharp_{\xi} \geq 0 \qquad \mathrm{and} \qquad \int_{\RM}f^\sharp_{\xi}(X) \; \d X=1.
\end{equation}
\end{lemma}
\begin{proof}
Let us fix $\xi \in \SM^{n-1}$. Observe that for all $X \in \RM$
\begin{equation}
f^{\sharp}_{\xi}(X)=f^{\sharp}(X,\xi)=\int_{\RM^n}f(x)\,\delta(X-\xi \cdot x) \;\d x \geq 0.
\end{equation}
Moreover,
\begin{equation}
\int_{\RM}f^\sharp_{\xi}(X) \; \d X=\int_{\RM} \d X \int_{\RM^n} f(x)\, \delta(X-\xi \cdot x) \; \d x=\int_{\RM^n}f(x) \; \d x=1.
\end{equation}
\end{proof}

The Radon transform has a peculiar behavior under translation.
Given a displacement $a\in\RM^n$, consider the translation operator defined by
\begin{equation}
(\tau_a f) (x) = f(x-a).
\end{equation}
One has the following property
\begin{lemma}
\label{lm:translation}
Let $f \in  \mathcal{S}(\RM^n)$ and $a\in\RM^n$. Then for any $\xi\in\SM^{n-1}$
\begin{equation}
(\tau_a f)^\sharp_\xi=\tau_{\xi\cdot a} f^\sharp_{\xi}.
\end{equation}
\end{lemma}
\begin{proof}
By setting $y= x- a$ we get
\begin{eqnarray}
(\tau_a f)^\sharp_\xi (X) &=&
\int_{\RM^n} f(x-a)\, \delta(X-\xi\cdot x)\;\d x =\int_{\RM^n} f(y)\, \delta(X-\xi\cdot y-\xi\cdot a)\;\d y \nonumber\\
&=& f^\sharp_\xi (X-\xi\cdot a). \nonumber
\end{eqnarray}
\end{proof}

The following lemma, which is the infinitesimal version of Lemma~\ref{lm:translation}, shows how the Radon transform and its dual intertwine the Laplacian on $\RM^n$, $\Delta=\sum_i \partial_{x_i}^2$ with the Laplacian on $\RM$, $\partial^2_X$.
\begin{lemma}
Let $f \in  \mathcal{S}(\RM^n)$ and $g \in  \mathcal{S}(\PM^n)$. Then
\begin{equation}
(\Delta f)^{\sharp}=\frac{\partial^2 f^{\sharp}}{\partial X^2} \qquad\text{and} \quad \left(\frac{\partial^2 g}{\partial X^2}\right)^{\flat}=\Delta g^{\flat}.
\label{eq:intertwine}
\end{equation}
\end{lemma}
\begin{proof}
By Lemma~\ref{lm:translation} we get
\begin{equation}
(\tau_{a}f)^{\sharp}(X,\xi)=f^{\sharp}(X-\xi\cdot a, \xi).
\end{equation}
Therefore, by considering an orthonormal basis $(e_1,\dots,e_n)$ of $\RM^n$,
for all $X\in \RM^n$ and $\xi \in \SM^{n-1}$ one has that
\begin{equation}
\left(\frac{\partial f}{\partial x_{i}}\right)^{\sharp}(X,\xi)= \lim_{h \to 0} \left(\frac{f-\tau_{h e_{i}}f}{h}\right)^{\sharp} (X,\xi)=\xi_i \frac{\partial f^{\sharp}}{\partial X}(X,\xi)
\end{equation}
and thus
\begin{eqnarray}
\left(\frac{\partial^2 f}{\partial x_{i}^2}\right)^{\sharp} (X,\xi) =  \xi_i\frac{\partial}{\partial X} \left(\frac{\partial f}{\partial x_{i}}\right)^{\sharp} (X,\xi)
=\xi_i^2 \frac{\partial^2 f^{\sharp}}{\partial X^2}(X,\xi).
\end{eqnarray}
Therefore, since $\xi \in \SM^{n-1}$, by summing over $i$ we have the first equality in  (\ref{eq:intertwine}).

As for the second relation, observe that
\begin{equation}
g^{\flat}(x)=\int_{\RM\times \SM^{n-1}} g(X,\xi)\, \delta(X-\xi \cdot x)\; \d X \d\xi=\int_{\SM^{n-1}} g(\xi \cdot x,\xi)\; \d\xi,
\end{equation}
whence
\begin{equation}
\frac{\partial g^\flat}{\partial x_i}(x)=\int_{\SM^{n-1}} \frac{\partial g}{\partial X}(\xi \cdot x,\xi)\, \xi_i\; \d\xi,
\qquad
\frac{\partial^2 g^\flat}{\partial x_i^2}(x)=\int_{\SM^{n-1}} \frac{\partial^2 g}{\partial X^2}(\xi \cdot x,\xi) \, \xi_i^2\; \d\xi .
\end{equation}
Therefore,
\begin{equation}
\label{ }
\Delta g^\flat(x)=\sum_{i=1}^n \frac{\partial^2 g^\flat}{\partial x_i^2}(x)=\int_{\SM^{n-1}} \frac{\partial^2 g}{\partial X^2}(\xi \cdot x,\xi) \; \d\xi =\left(\frac{\partial^2 g}{\partial X^2}\right)^\flat(x).
\end{equation}
\end{proof}

The Radon transform is closely connected with the Fourier transform.  Let us recall some basic facts about the Fourier transform.
For any function $g\in\mathcal{S}(\RM^m)$ with $m\in\NM_*$, we define its Fourier transform
as
\begin{equation}\label{fourier transform}
\hat{g}(k)=\int_{\RM^m} g(x)\, \e^{-\i k \cdot x}\; \d x,
\end{equation}
with $k \in \RM^m$.  It is not difficult to prove that the Fourier transform is a bijection of $\mathcal{S}(\RM^m)$ onto itself, and its inverse reads
\begin{equation}\label{inv fourier transform}
\check{g}(x)=\int_{\RM^m} g(k)\, \e^{\i k \cdot x}\; \frac{\d k}{(2\pi)^m}.
\end{equation}
In short, for any  $g\in\mathcal{S}(\RM^m)$, one has that $\check{\hat{g}}=\hat{\check{g}}=g$. Using standard duality arguments the Fourier transform can be extended to the space of temperated distributions $\mathcal{S}'(\RM^m)$ (and thus also to its subspace of square summable functions $L^2(\RM^n)$, where it can also be proved to be unitary). In particular, we recall that $\hat{\delta}=1$ and thus
\begin{equation}
\delta(x)=\int_{\RM^n} \e^{\i k\cdot x}\; \frac{\d k}{(2\pi)^n}.
\end{equation}

We will show that the $n$-dimensional Fourier transform of a function $f$ on $\RM^n$ is the $1$-dimensional Fourier transform of its tomogram $f_\xi^\sharp$ for all $\xi\in\SM^{n-1}$. This result is known as the \emph{Fourier slice theorem}.
\begin{lemma}\label{R-F}
Consider $f \in \mathcal{S}(\RM^n)$, $\tau \in \RM$ and $\xi \in \SM^{n-1}$. Then
\begin{equation}\label{Radon-Fourier}
\hat{f}(\tau \xi)=\widehat{ f^\sharp_\xi} (\tau)
\end{equation}
\end{lemma}
\begin{proof}
The thesis follows immediately from the definitions. Indeed
\begin{eqnarray}
\widehat{ f^\sharp_\xi } (\tau) &=& \int_{\RM} \d X\; \e^{-\i\tau X}\, f^{\sharp}_\xi (X)
=\int_{\RM}\d X\; \e^{-\i\tau X} \int_{\RM^n} \d x\;  f(x)\, \delta(X-\xi \cdot x) \nonumber \\
& = & \int_{\RM^n} \d x\; f(x) \int_{\RM} \d X\;   \e^{-\i\tau X}\, \delta(X-\xi \cdot x)
= \int_{\RM^n} \d x\; f(x)\, \e^{-\i\tau \xi \cdot x}
= \hat{f}(\tau \xi). \quad
\end{eqnarray}
\end{proof}

Now we will address the problem of the range of the Radon transform with domain $\mathcal{S}(\RM^n)$. The following lemma is crucial.
\begin{lemma}\label{homogeneous}
If $f \in \mathcal{S}(\RM^n)$, then for any integer $k \in \NM$ and any $\xi=(\xi_1,\dots,\xi_n)\in\SM^{n-1}$, the integral of its tomogram,
\begin{equation}
I_f(\xi)=\int_{\RM}f^{\sharp}_\xi(X)\, X^k \; \d X,
\end{equation}
is a homogeneous polynomial of degree $k$ in the Cartesian components $\xi_1, \dots, \xi_n$.
\end{lemma}
\begin{proof}
This is immediate from the relation
\begin{equation}
\int_{\RM}f^{\sharp}(X,\xi) \, X^k \; \d X = \int_{\RM} \d X \, X^k \int_{\RM^n} \d x\; f(x) \, \delta(X-\xi \cdot x)
= \int_{\RM^n} f(x) \, (\xi \cdot x)^k\; \d x.
\end{equation}
\end{proof}

In accordance with this lemma we define the space $\mathcal{S}_H(\PM^n)$ as follows: $g \in \mathcal{S}_H(\PM^n)$ if and only if $g \in \mathcal{S}(\PM^n)$ and for any integer $k \in \NM$ the integral
\begin{equation}
I_g(\xi)=\int_{\RM} g(X,\xi)\, X^k \; \d X
\end{equation}
is a homogeneous polynomial of degree $k$ in $\xi_1, \dots, \xi_n$.

\begin{theorem}
The Radon transform is a linear one-to-one mapping of $\mathcal{S}(\RM^n)$ onto $\mathcal{S}_H(\PM^n)$.
\label{th:Radonbijection}
\end{theorem}
The proof uses the well-known fact that the Fourier transform maps $\mathcal{S}(\RM^n)$ onto itself and Lemma~ \ref{homogeneous}. For more details see \cite{Helgason}. In the next section we will give a constructive proof, by explicitly  exhibiting the inverse mapping.

\section{Inversion formula}
Now we want to obtain the inversion formula, namely we want to prove that one can recover a function $f$ on $\RM^n$ if all its tomograms $(f^{\sharp}_{\xi})_{\xi \in \SM^{n-1}}$ are known. In order to get this result we need three preliminary lemmata.
\begin{lemma}\label{sharp flat}
Let $f \in \mathcal{S}(\RM^n)$ and $V(x)=1/|x|$, $x \in \RM^n$, $x\neq0$ . Then
\begin{equation}
(f^{\sharp})^{\flat}(x)=a_{n}\, f\ast V,
\end{equation}
where $a_n$ depends only on the dimension $n$, and
$*$ denotes the convolution product,
\begin{equation}
(f\ast g)(x) =\int_{\RM^n} f(y)\, g(x-y) \; \d y  .
\end{equation}
\end{lemma}
\begin{proof}
Observe that
\begin{eqnarray}\label{double}
(f^{\sharp})^{\flat}(x) & = & \int_{\SM^{n-1}}\d\xi \; f^{\sharp}(\xi \cdot x, \xi)
= \int_{\SM^{n-1}}\d\xi \int_{\RM^n} \d y\; f(y)\, \delta(\xi \cdot x- \xi \cdot y) \nonumber \\
                    & = & \int_{\SM^{n-1}}\d\xi \int_{\RM^n} \d y\; f(y) \int_{\RM} \frac{\d s}{2\pi}\; \e^{\i s \xi \cdot(x-y)}.
\end{eqnarray}
Now, with the substitution $x-y=|x-y|\omega$, $\omega \in \SM^{n-1}$, and $t=s|x-y|$, the integral (\ref{double}) becomes
\begin{equation}
\int_{\SM^{n-1}}\d\xi \int_{\RM^n} \d y\; f(y) \int_{\RM} \frac{\d t}{2\pi}\;\frac{1}{|x-y|}\,\e^{\i t \xi \cdot \omega}
=  a_{n}\int_{\RM^n} \d y\; f(y)\, V(x-y),
\end{equation}
where
\begin{equation}
a_{n}=\int_{\RM} \frac{\d t}{2\pi} \int_{\SM^{n-1}} \d\xi \; \e^{\i t\xi \cdot \omega}.
\label{eq:andef}
\end{equation}
By symmetry, it is apparent that $a_n$ is independent of $\omega$.
\end{proof}

\begin{figure}[t]
\includegraphics[width=0.6\textwidth]{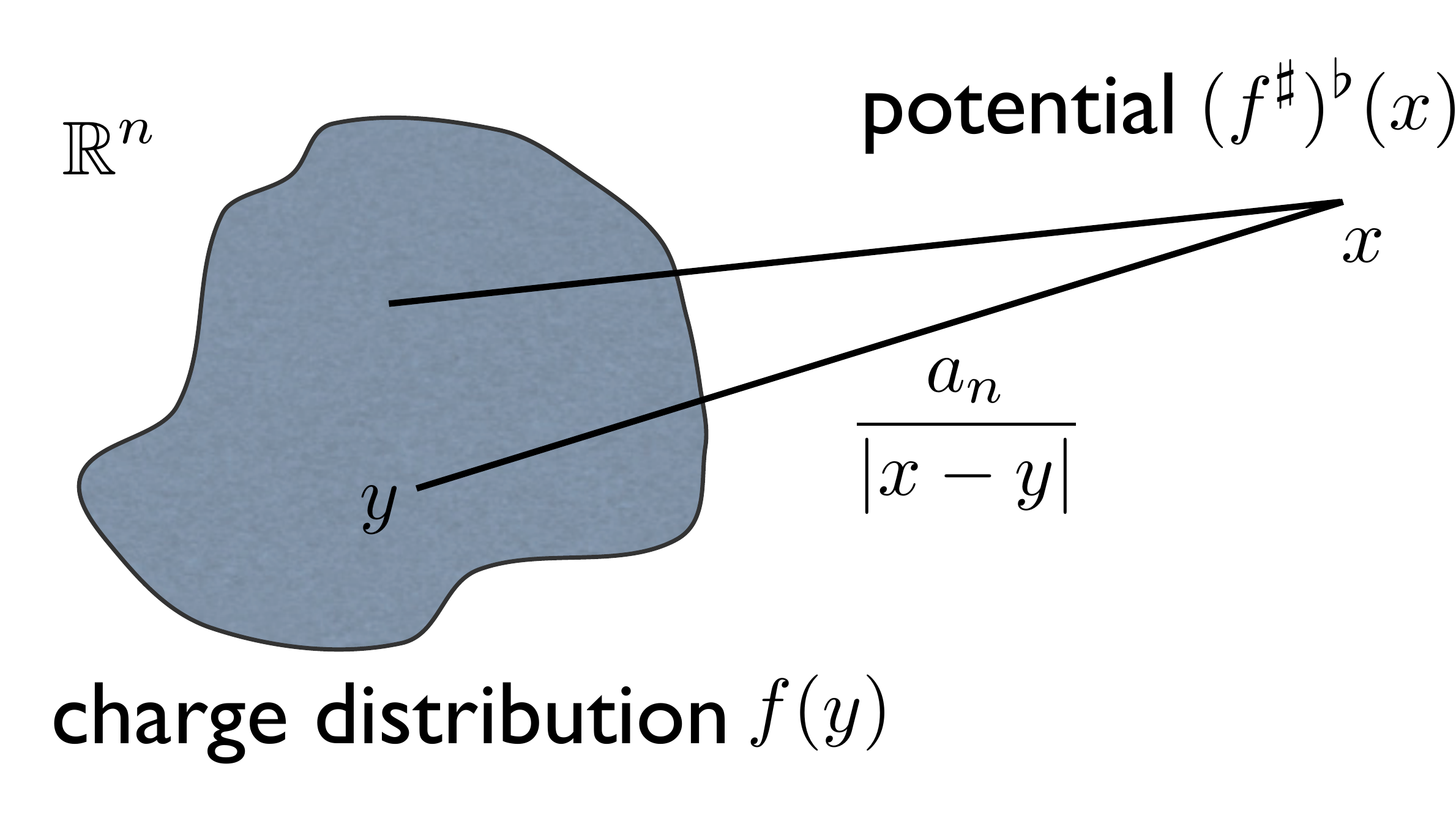}
\caption{$(f^{\sharp})^{\flat}(x)$ is the potential at $x$ generated by the charge distribution $f$.}
\label{fig:potential}
\end{figure}

The physical meaning of Lemma~\ref{sharp flat} is the following: if $f$ is a charge distribution, then the potential at the point $x$ generated by that charge is exactly $(f^{\sharp})^{\flat}(x)$, see Fig.~\ref{fig:potential}. Notice, however, that the potential of a point charge scales always as the inverse distance \emph{independently} of the dimension $n$, and thus it is Coulomb only for $n=3$. The only dependence on $n$ is in the strength of the elementary charge $a_n$. This fact will be crucial in the following: we will see that the associated Poisson equation involves an $n$-dependent (fractional) power of the Laplacian.

\begin{lemma}\label{lm:hatV}
Let $V(x)=1/|x|$, $x \in \RM^n$, $x\neq0$ . Then its Fourier transform is
\begin{equation}
\hat{V}(k)=\frac{b_n}{|k|^{n-1}}.
\end{equation}
where $b_n$ depends only on $n$
\end{lemma}
\begin{proof}
One gets
\begin{equation}
\hat{V}(k) =  \int_{\RM^n}\d x\; \frac{1}{|x|}\, \e^{-\i k \cdot x} = \int_{\RM^n}\d x\;\frac{1}{|x|}\, \e^{-\i |k|\omega \cdot x}
= \int_{\RM^n}\frac{\d y}{|k|^{n-1}}\;\frac{1}{|y|}\, \e^{-\i\omega \cdot y}=\frac{b_n}{|k|^{n-1}},
\end{equation}
where we have set $k=|k| \omega$ with $\omega\in\SM^{n-1}$, $y=|k| x$, and
\begin{equation}
b_{n}=\int_{\RM^n} \d x\; \frac{\e^{-\i \omega \cdot x}}{|x|}.
\end{equation}
Moreover, $b_n$ is independent of $\omega$ by symmetry.
\end{proof}

\begin{lemma}\label{lm:anbn}
Let $a_n$ and $b_n$ as in Lemmata~\ref{sharp flat} and \ref{lm:hatV}. Then
\begin{equation}
a_n b_{n}=2(2\pi)^{n-1}.
\label{an bn1}
\end{equation}
\end{lemma}
\begin{proof}
Since the Fourier transform is a bijection, we use the trick of going back and forth.
From Lemma~ \ref{lm:hatV} we have
\begin{eqnarray}
V(x)=\frac{1}{|x|} =  \check{\hat{V}}(x)=
\int_{\RM^n}\hat{V}(k)\, \e^{\i k \cdot x} \; \frac{\d k }{(2\pi)^n}\
=  \int_{\RM^n}\frac{b_n}{|k|^{n-1}}\, \e^{\i k \cdot x}\;  \frac{\d k}{(2\pi)^n}.
\label{an bn}
\end{eqnarray}
By setting  $x=|x|\omega$ and $p=|x|k= t\xi$, with $\omega, \xi \in \SM^{n-1}$ and $t\in\RM^+$, Eq.~(\ref{an bn}) becomes
\begin{eqnarray}
\frac{1}{|x|}
& = &\frac{b_n}{(2\pi)^n}\int_{\RM^n}\frac{\d k}{|k|^{n-1}}\, \e^{\i k \cdot \omega |x|}
= \frac{b_n}{(2\pi)^n|x|}\int_{\RM^n}\frac{\d p}{|p|^{n-1}}\, \e^{\i p \cdot \omega} \nonumber\\
& = & \frac{b_n}{(2\pi)^n|x|}
\int_{\RM^+} \d t \; t^{n-1}\int_{\SM^{n-1}}\frac{\d\xi}{t^{n-1}}\; \e^{\i t\xi \cdot \omega} \nonumber\\
& = & \frac{b_n}{2(2\pi)^{n-1}|x|}\int_{\RM}\frac{\d t}{2\pi} \int_{\SM^{n-1}}\d\xi\; \e^{\i t\xi\cdot\omega}
= \frac{a_n b_n}{2(2\pi)^{n-1}}\, \frac{1}{|x|},
\end{eqnarray}
where in the last equality we used Eq.~(\ref{eq:andef}) of Lemma~\ref{sharp flat}.
Therefore,
\begin{equation}
\frac{a_n b_n}{2(2\pi)^{n-1}}=1
\end{equation}
and the thesis follows.
\end{proof}

Now we are ready to prove the inversion formula for the Radon transform.
\begin{theorem}
Let $f \in \mathcal{S}(\RM^n)$. Then
\begin{equation}\label{inversion formula0}
f(x)= \frac{1}{2(2\pi)^{n-1}}(-\Delta)^{\frac{n-1}{2}} (f^\sharp)^\flat(x)
\end{equation}
where $(-\Delta)^\alpha$, with $\alpha>0$, is a pseudodifferential operator whose action is
\begin{equation}
((-\Delta)^\alpha f) (x)= \int_{\RM^n} |k|^{2\alpha}\, \hat{f}(k)\,  \e^{\i k \cdot x}\; \frac{\d k}{(2\pi)^n}.
\label{eq:fractionalLap}
\end{equation}
\end{theorem}
\begin{proof}
First of all observe that if $n=1$ statement (\ref{inversion formula}) is trivial, so assume that $n >1$. Let $p>0$, we get by Lemma~ \ref{sharp flat}
\begin{equation}
(-\Delta)^p V(x)= \int_{\RM^n}\frac{\d k}{(2\pi)^n}\; |k|^{2p}\, \hat{V}(k)\,  \e^{\i k \cdot x}
= b_n \int_{\RM^n}\frac{\d k}{(2\pi)^n}\; |k|^{2p-n+1}\,  \e^{\i k \cdot x} ,
\end{equation}
so that, if $2p=n-1$, we have
\begin{equation}
(-\Delta)^{\frac{n-1}{2}} V(x)=b_n\delta(x).
\end{equation}
Therefore, again by Lemma \ref{sharp flat}, we have
\begin{equation}
(-\Delta)^{\frac{n-1}{2}}(f^{\sharp})^\flat (x) =  a_n \left( f \ast (-\Delta)^{\frac{n-1}{2}}V \right)(x)
= a_n b_n (f \ast \delta)(x)
= a_n b_n f(x).
\end{equation}
From (\ref{an bn1}) we get the inversion formula (\ref{inversion formula}).
\end{proof}
Equation~(\ref{inversion formula0}) says that, modulo the final action of $(-\Delta)^{\frac{n-1}{2}}$, the function $f$ can be recovered from its Radon transform $f^\sharp$ by the application of the dual mapping $\flat$: first one integrates over the set of points in a hyperplane and then one integrates over the set of hyperplanes passing through a given point. Explicitly we get
\begin{equation}\label{inversion formula}
f(x) =\frac{1}{2(2\pi)^{n-1}}(-\Delta)^{\frac{n-1}{2}}\int_{\SM^{n-1}}   f^{\sharp}(\xi \cdot x, \xi)\; \d\xi,
\end{equation}
which has the following remarkable interpretation. Note that if one fixes a direction $\xi \in \SM^{n-1}$, then the function $f^{\sharp}(\xi \cdot x,\xi)$ is a plane wave, i.e.\ a function constant on each plane perpendicular to $\xi$. Therefore, Eq.~(\ref{inversion formula}) gives a representation of $f$ in terms of a continuous superposition of plane waves. A well-known analogous decomposition is given by Fourier transform. When $n=3$, one recovers the inversion formula originally found by Radon \cite{Rad1917}
\begin{equation}
f(x) =-\frac{1}{8 \pi^{2}} \Delta \int_{\SM^{2}}   f^{\sharp}(\xi \cdot x, \xi)\; \d\xi.
\end{equation}

\subsection{Localization}
Now we consider the following question on the localization properties of the Radon transform: given a function $f$ can we reconstruct it in a compact region $\Omega$,  using only the tomograms of $f$ on  hyperplanes $\lambda$ that pass through  $\Omega$? By looking at the inversion formula (\ref{inversion formula}) one would naively think that the answer be positive. By following a nice argument given by John \cite{John}, we will show that in fact this cannot be true in two (and, in general, even) dimensions. The reason is the well-known fact that waves cannot be localized in two dimensions: a pebble dropped in a pond produces ripples that  propagate along expanding disks. On the other hand, in one dimension waves remain localized: a local perturbation of the pressure in an organ pipe propagates by translating its support. This is a fortunate accident that allows us to enjoy a Bach's fugue without hearing an endless annoying echo of each tune.

\begin{figure}[t]
\includegraphics[width=0.45\textwidth]{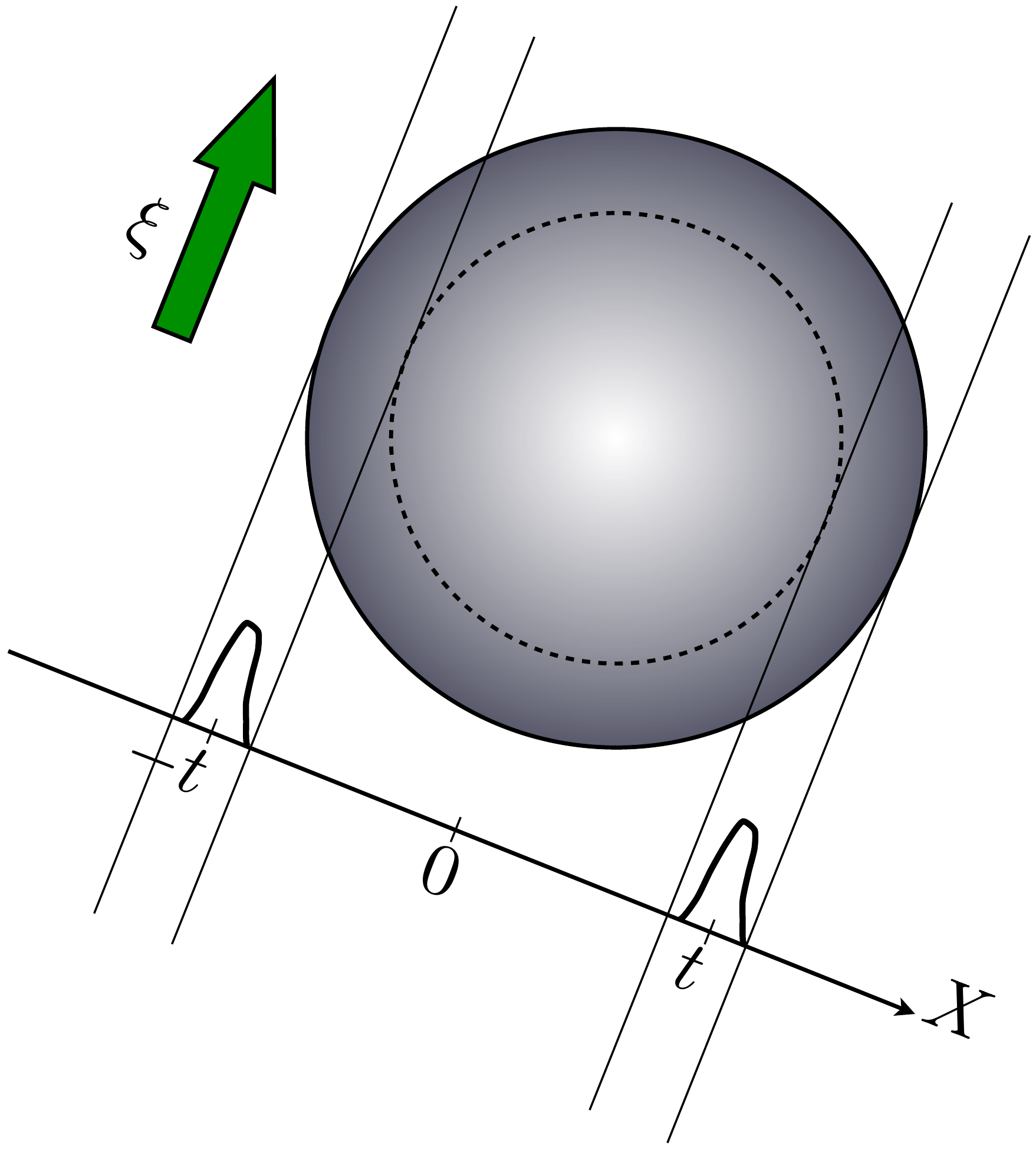}
\caption{A function $f$ supported in a ball and its tomogram in the direction $\xi$.}
\label{fig:localization}
\end{figure}

Consider the following Cauchy problem for the wave equation in $\RM^2$: 
\begin{equation}\label{wave equation}
\partial_{t}^2 f-\Delta f=0, \qquad
f(0,x)=g(x), \qquad
\partial_{t}f(0,x)=0,
\end{equation}
and assume that  the initial condition $g$ has its support in the ball $B(0,\epsilon)=\{|x|<\epsilon\}$ with center $0$ and radius $\epsilon>0$.
Let us fix a time $t>\epsilon$. The solution to this equation  is given by the celebrated Poisson's formula \cite{Evans}
\begin{equation}\label{soluz.wave eq 2d}
f(t,x)=\frac{1}{2\pi t}\int_{B(x,t)} \frac{g(y)+\nabla g(y) \cdot (y-x)}{\sqrt{t^2-|y-x|^2}}\; \d y ,
\end{equation}
where $\nabla$ denotes the gradient in $\RM^2$. Therefore the support of $f(t,\cdot)$ is a subset of $B(0,t+\epsilon)$.
On the other hand, if we consider the Radon transform of (\ref {wave equation}) we get, for every $\xi \in \SM^{n-1}$, a
Cauchy problem for the wave equation in $\RM$
\begin{equation}\label{wave equation RT}
\partial_{t}^2 f^{\sharp}_{\xi}- \partial_{X}^2 f^{\sharp}_{\xi}=0 ,
\qquad
f^{\sharp}_{\xi}(0,X)=g^{\sharp}_{\xi}(X) ,
\qquad
\partial_{t}f^{\sharp}_{\xi}(0,X)=0.
\end{equation}
Observe that $g^{\sharp}_{\xi}$ has its support in $(-\epsilon,\epsilon)$. The solution of Eq.~(\ref{wave equation RT}) is
\begin{equation}
f^{\sharp}_{\xi}(t,X)=\frac{1}{2}\left(g^{\sharp}_{\xi}(X-t)+g^{\sharp}_{\xi}(X+t)\right),
\end{equation}
thus the function $f^{\sharp}_{\xi}(t,\cdot)$ has its support in $(-t-\epsilon,-t+\epsilon)\cup (t-\epsilon,t+\epsilon)$.
So we see that if one wants to reconstruct the function $f(t,\cdot)$ in the set $B(0,t-\epsilon)$ one cannot consider only the tomograms on the lines that intersect $B(0,t-\epsilon)$, because they are all equal to zero.
See Fig.~\ref{fig:localization}.
Indeed, from $f^\sharp(t, X,\xi)=0$ for $|X|<t-\epsilon$ we have that
\begin{equation}
(f^\sharp)^\flat(t,x)= \int_{\SM}  f^\sharp(t,\xi \cdot x,\xi)\; \d\xi =0, \qquad \text{for}\quad |x|< t-\epsilon.
\end{equation}
On the other hand, the inversion formula (\ref{inversion formula}) for $n=2$ yields
\begin{equation}
f(t,x)  = \frac{1}{4\pi}(-\Delta)^{1/2}(f^\sharp)^\flat(t,x) \neq 0, \qquad \text{for}\quad |x|< t-\epsilon.
\end{equation}
 The seeming contradiction derives from the fact that the pseudodifferential operator $(-\Delta)^{1/2}$ is nonlocal. In fact, by (\ref{eq:fractionalLap}) it is an integral operator that depends also on the values of $(f^\sharp)^\flat(t,x)$ for $|x|> t-\epsilon$. Of course, this is a peculiar behavior of the even dimension. For odd $n$ in the inversion formula (\ref{inversion formula}) we have an integer power of the Laplacian that is a \textit{bona fide} differential operator and thus is local.

\section{$M^2$-transform}\label{sec:M^2}
In this section we introduce another covering of the homogeneous space $\PM^n$ that uses an additional parameter, which is absent in the classical tomographic framework of Radon. This choice of parametrization of $\PM^n$ gives an easier inversion formula, although this approach is completely equivalent to that considered in the previous section.

We define the following covering
\begin{equation}
\RM \times \RM^n_*
\to \PM^{n}, \qquad (X,\mu) \mapsto \lambda,
\end{equation}
where the Cartesian equation for $\lambda$ is
\begin{equation}
\lambda=\{x\in\RM^n \,|\, X-\mu \cdot x=0\}.
\end{equation}
Of course the additional parameter with respect to Radon is $|\mu|$.
Now, if $f \in \mathcal{S}(\RM^n)$ we can define the $M^2$-transform of $f$ as follows
\begin{equation}\label{M^2transform}
f^{M^2}(X,\mu)=\int_{\RM^n}  f(x)\, \delta(X-\mu \cdot x)\; \d x=\langle\delta(X-\mu \cdot x)\rangle_f.
\end{equation}
This tomographic mapping was introduced by Man'ko and Marmo, who
gave seminal contribution towards its significance
\cite{MarmoPhysScr,MarmoJPA,MarmoPL}.
Also for this transform we can define the tomograms, precisely for every $\mu \in \RM^n_*$ we define $f^{M^2}_{\mu}(X):=f^{M^2}(X,\mu)$. It easy to check that also this transform maps a probability density on $\RM^n$ into a family of probability density functions $\RM$, namely if $f \in \mathcal{S}(\RM^n)$ is a positive and normalized function, then for all $\mu \in \RM^n_*$, $f^{M^2}_{\mu}$ is positive and normalized.

Let us now consider the relations between Radon and $M^2$ transform.
Obviously, from
$f^{M^2}(X,\mu)$ one can immediately recover $f^{\sharp}(X,\xi)$
by setting $\mu=\xi\in \SM^{n-1}$:
\begin{eqnarray}
f^{\sharp}(X,\xi)=f^{M^2}(X,\xi).
\label{eq:ff0}
\end{eqnarray}
However, notice that, although $f^{\sharp}$ is the restriction of
$f^{M^2}$ to the unit sphere $\SM^{n-1}$, there is actually a
bijection between the two transforms. Therefore, they carry exactly
the same information. Indeed, since the Dirac distribution is
positive homogeneous of degree $-1$, i.e.\ $\delta(\alpha x)=
|\alpha|^{-1} \delta(x)$, for every $\alpha\neq 0$, one gets from
Eq.\ (\ref{M^2transform})
\begin{eqnarray}
f^{M^2}(X,\mu)=\frac{1}{|\mu|}f^{M^2}\left( \frac{X}{|\mu|},\frac{\mu}{|\mu|}
 \right),
\end{eqnarray}
for $\mu \neq 0$. In words, the tomogram $f^{M^2}(X,\mu)$ at a
generic point $\mu \in \RM^n_*$ is completely determined by the
tomogram at $\mu/|\mu| \in \SM^{n-1}$. But the latter is nothing but
the Radon transform, by Eq.\ (\ref{eq:ff0}). Therefore, we get the
bijection
\begin{eqnarray}
f^{\sharp}(X,\xi)=f^{M^2}(X,\xi), \qquad
f^{M^2}(X,\mu)=\frac{1}{|\mu|}f^{\sharp}\left( \frac{X}{|\mu|} ,\frac{\mu}{|\mu|}
 \right).
\label{eq:ff33}
\end{eqnarray}

Finally, let us look at an inversion formula as in the case of the Radon transform.
\begin{theorem}
Let $f \in \mathcal{S}(\RM^n)$. Then
\begin{equation}\label{M^2 inversion formula}
f(x)=\int_{\RM^{n+1}} f^{M^2}(X,\mu)\, \e^{\i(X-\mu \cdot x)} \; \frac{\d X \d\mu}{(2\pi)^n}.
\end{equation}
\end{theorem}
\begin{proof}
Observe that
\begin{eqnarray}
& & f^{M^2}(X,\mu)  =   \int_{\RM^n} \d x\; f(x)\, \delta(X-\mu \cdot x)
= \int_{\RM^n} \d x\; f(x) \int_{\RM}\frac{\d \tau}{2\pi}\; \e^{\i\tau(X-\mu \cdot x)} \nonumber\\
&  & \quad = \int_{\RM^{n+1}}\frac{\d y\, \d \tau}{(2\pi)^{n+1}}\; \e^{\i(\tau X-\mu \cdot y)}\,  f \left(\frac{y}{\tau}\right)\frac{(2\pi)^{n}}{|\tau|^n}
=  \int_{\RM^{n+1}}\frac{\d y\, \d \tau}{(2\pi)^{n+1}}\; \e^{\i(\tau X-\mu \cdot y)} g(\tau,y),\qquad
\end{eqnarray}
where $g(\tau,y)= (2\pi)^{n}f(y/\tau)/ |\tau|^n $. Therefore,
$f^{M^2}(X,\mu) = \check{g}(X,-\mu)$,
and so
\begin{equation}
g(\tau,x)=\widehat{f^{M^2}}(\tau,-x).
\end{equation}
Finally,
\begin{equation}
f(x)=\frac{g(-1,-x)}{(2\pi)^{n}}=\frac{\widehat{f^{M^2}}(-1,x)}{(2\pi)^{n}}=\int_{\RM^{n+1}}\frac{\d X\,  \d\mu}{(2\pi)^n}\;f^{M^2}(X,\mu)\, \e^{\i(X-\mu \cdot x)}.
\end{equation}
\end{proof}

Summarizing, the inversion formulae for the transforms (\ref{inversion formula}) and
(\ref{M^2 inversion formula}) read
\begin{eqnarray}
f(x)  &=&  \frac{1}{2(2 \pi)^{n-1}}
(-\Delta)^{\frac{n-1}{2}}\int_{\SM^{n-1}} f^{\sharp}(\xi \cdot x,\xi)\; \d\xi,
\label{eq:geninvrad}
\\
f(x) &=&  \int_{\mathbb{R}^{n+1}} f^{M^2}(X,\mu)\,
\e^{\i(X-\mu \cdot x)}\; \frac{\d X \d\mu}{(2\pi)^n},
\label{eq:geninvsym}
\end{eqnarray}
respectively. While formula (\ref{eq:geninvsym}), which is nothing
but a Fourier transform, is quite easy to handle, the inversion
formula (\ref{eq:geninvrad}) is in general very hard to tackle,
especially for even $n$, due to the presence of a fractional
Laplacian. For an explicit example, see \cite{m2}.

\section{Group-theoretic framework}
The general problem posed by the Radon transform is that of determining a function on a manifold by means of its integrals over a certain family of submanifolds. Here, following Helgason \cite{Helgason} we will construct a quite general framework for the above problem.

Let us start from the simple observation that the Euclidean group $E(n)$ acts transitively in a natural way both on the affine space $\RM^n$ and on the hyperplane space $\PM^n$. The Euclidean group is the group of isometries associated with the Euclidean metric  and is the semidirect product  of the orthogonal group $O(n)$ of rotations, extended by the abelian group of translations $\RM^n$, namely  $E(n)=O(n)\ltimes\RM^n$.

Fix an origin $0$ in $\RM^n$ (a reference point) and an origin $\lambda_0$ in $\PM^n$ (a reference hyperplane). It will be convenient to choose an hyperplane $\lambda_0$ passing through the origin $0$.
Transitivity means that any element of $\RM^n$ ($\PM^n$) can be obtained by an Euclidean motion of the origin $0$ ($\lambda_0$), namely
\begin{eqnarray}
& &\forall x \in\RM^n :\; x= g \cdot 0, \quad \text{for some}\; g\in E(n), \nonumber\\
& &\forall \lambda \in\PM^n :\; \lambda= \gamma \cdot \lambda_0, \quad \text{for some}\; \gamma\in E(n),
\end{eqnarray}
where the dot $\cdot$ denotes the action of $E(n)$. The action is not free, i.e.\ one can have $g\cdot x = h\cdot x$ for some $x\in \RM^n$ and $g,h\in E(n)$ with $g\neq h$. Therefore, let us consider the isotropy subgroup of the origin $0$ (also called the stabilizer subgroup) as the set of all elements in $E(n)$ that fix $0$:
\begin{equation}
E(n)_{0}=\left\{ g \in E(n) \,|\, g\cdot 0=0  \right\} = O(n).
\end{equation}
It is given by the rotation group (around $0$). On the other hand,
the isotropy subgroup of the reference hyperplane $\lambda_{0}$ passing through the origin is
\begin{equation}
E(n)_{\lambda_0}=\left\{g \in E(n) \left| \right. g \cdot \lambda_{0}=\lambda_{0}  \right\} =\ZM_2 \times E(n-1) ,
\end{equation}
where $E(n-1)$ is the Euclidean group of the $(n-1)$-dimensional affine subspace $\lambda_0$ (translations and rotations in $\lambda_0$), while $\ZM_2\simeq O(1)$ is the group whose elements are the identity and the reflection in $\lambda_0$. See Fig. \ref{fig:isotropy}.
\begin{figure}[t]
\includegraphics[width=0.45\textwidth]{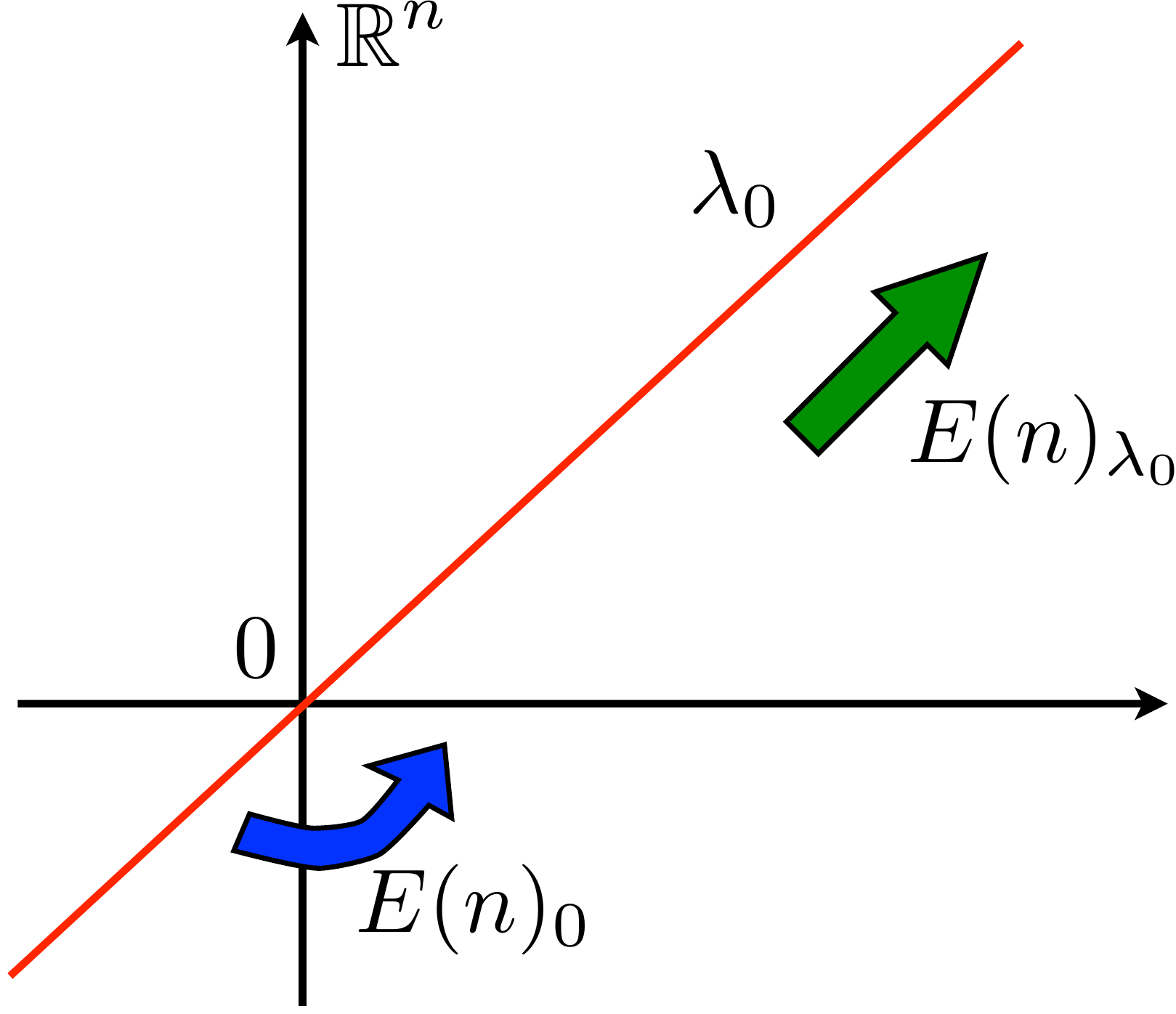}
\caption{Origins $0$ and  $\lambda_0$ and their isotropy groups.}
\label{fig:isotropy}
\end{figure}
Now, given an element $x\in \RM^n$, by transitivity we get $x=g\cdot 0$, for some $g\in E(n)$. Therefore we can identify $x$ with the left coset $g E(n)_0 = g  O(n)$. Analogously, for all $\lambda\in\PM^n$ we get the identification $\lambda \simeq \gamma E(n)_{\lambda_0}= \gamma (\ZM_2 \times E(n-1))$, for some $\gamma\in E(n)$. Therefore, we get that
\begin{eqnarray}
\RM^n &\simeq& E(n)/E(n)_{0} =E(n)/O(n),  \nonumber\\
\PM^n  &\simeq& E(n)/E(n)_{\lambda_0}=E(n)/(\ZM_2 \times E(n-1)),
\end{eqnarray}
are homogeneous spaces of $E(n)$.

We want to describe the incidence relation in this framework.  Namely, given an hyperplane $\lambda\in\PM^n$ and a point $x \in \RM^n$, when $x$ lies in $\lambda$, that is $x \in \lambda$? It is not difficult to see that this is true if and only if their corresponding cosets, considered as subsets of $E(n)$, have a point in common.
In equations, given $x \simeq g O(n)$ and $\lambda \simeq \gamma (\ZM_2 \times E(n-1))$, one has
\begin{equation}
x \in \lambda \qquad \Leftrightarrow \quad
g O(n) \cap \gamma (\ZM_2 \times E(n-1)) \neq \emptyset .
\label{eq:incidence0}
\end{equation}
Indeed, one gets
\begin{eqnarray}
g\cdot 0 \in \gamma \cdot\lambda_0 & \Leftrightarrow & g\cdot 0 = \gamma h \cdot 0, \quad \text{for some } h\in \ZM_2 \times E(n-1)  \nonumber\\
& \Leftrightarrow & g k = \gamma h, \quad \text{for some } k\in O(n)
\nonumber\\
& \Leftrightarrow &  g O(n) \cap \gamma (\ZM_2 \times E(n-1)) = \{ g k \}.
\label{eq:incidenceproof}
\end{eqnarray}

In the Cartesian product $\RM^n \times \PM^n$ we can  consider the following set, called the \emph{incidence relation},
\begin{equation}
I=\left\{(x, \lambda ) \in \RM^n \times \PM^n \left|\right. x \in \lambda \right\}.
\label{eq:incidence1}
\end{equation}
Consider now the restrictions to $I$ of the natural projections $p:I\subset \RM^n \times \PM^n\to\RM^n$ and $\pi: I\subset \RM^n \times \PM^n\to\PM^n$. Obviously, since every point of $\RM^n$ lies in some hyperplane of $\PM^n$ and every hyperplane passes through some point, we get $p(I)=\RM^n$ and $\pi(I)=\PM^n$, so that $p$ and $\pi$ are onto. Now, for any $x \in \RM^n$ and $\lambda \in \PM^n$ consider the sets
\begin{equation}
x^{\flat} = \left\{\lambda \in \PM^n  \left|\right. (x , \lambda) \in I \right\}, \qquad
\lambda^{\sharp} = \left\{ x \in \RM^n \left|\right. (x , \lambda) \in I \right\}.
\end{equation}
They represent  the set of all hyperplanes passing through $x$ and the set of all points lying in $\lambda$, respectively. The crucial fact is that the mappings
\begin{equation}
x\mapsto x^\flat=\pi\circ p^{-1}(x), \qquad \lambda\mapsto \lambda^\sharp = p\circ\pi^{-1}(\lambda),
\end{equation}
are in fact injective. Indeed, we obviously have
\begin{equation}
x= \bigcap_{\lambda \in x^\flat} \lambda, \qquad \lambda= \bigcup_{x\in \lambda^\sharp} x.
\end{equation}
Therefore, any point $x\in\RM^n$ can be uniquely identified with a set of hyperplanes
$x^\flat\subset\PM^n$ (those incident to $x$), and conversely any hyperplane $\lambda\in\PM^n$ can be uniquely identified with a set of points  $\lambda^\sharp\subset\RM^n$ (those incident to $\lambda$).
We say that $\RM^n$ and $\PM^n$ are \emph{homogeneous spaces in duality}. In this setting
we can rewrite the Radon transform and its dual as follows
\begin{equation}
f^{\sharp}(\lambda)=\int_{\lambda^{\sharp}} f(x) \; \d m(x)  , \qquad
\varphi^{\flat}(x)=\int_{x^{\flat}}  \varphi(\lambda) \; \d\mu(\lambda).
\label{eq:Radon+dual}
\end{equation}

Notice that from Eq.~(\ref{eq:incidenceproof}) we get for any pair $(x,\lambda)\in I$
\begin{equation}
(x,\lambda)\simeq \left(g O(n), g k (\ZM_2 \times E(n-1))\right) = \left(g l O(n), g k l (\ZM_2 \times E(n-1))\right),
\end{equation}
for some $k \in O(n)$   and for all $l\in L$, where
\begin{equation}
L=E(n)_0\cap E(n)_{\lambda_0} = O(n)\cap (\ZM_2 \times E(n-1)) = \ZM_2 \times O(n-1)
\end{equation}
is the subgroup that fixes both $0$ and $\lambda_0$.
Therefore, we get that also the incidence relation $I$ is a homogeneous space, namely
\begin{equation}
I\simeq E(n)/L = E(n)/(\ZM_2 \times O(n-1)).
\end{equation}
The Radon transform on the homogeneous spaces in duality $\RM^n$ and $\PM^n$ can be conveniently described  by the following double fibration:
\begin{equation}
\begin{diagram}
\node[2]{I\simeq E(n)/(\ZM_2 \times O(n-1))} \arrow{se,t}{\pi} \arrow{sw,t}{p}
\\
\node{\RM^n\simeq E(n)/O(n)}  \node[2]{\!\!\!\!\!\!\!\!\PM^n\simeq E(n)/(\ZM_2 \times E(n-1))}
\end{diagram}
\label{eq:doublefibration1}
\end{equation}

Now let us generalize the above construction. Let $G$ be a locally compact group, let $K$ and $H$ be two closed subgroups of $G$, and let $L=K\cap H$. Assume that $KH\subset G$ is closed.
Let us consider the following two left coset spaces
\begin{equation}
X=G/K \qquad \mathrm{and} \qquad \Lambda=G/H,
\end{equation}
and define the incidence relation
\begin{equation}
I=\left\{(x, \lambda ) \in X \times \Lambda \left|\right. x \cap \lambda \neq \emptyset \right\} .
\end{equation}
Define also for all $x \in X$ and $\lambda \in \Lambda$
\begin{equation}
x^{\flat}=\pi\circ p^{-1}(x)=\left\{\lambda \in \Lambda  \left|\right. (x , \lambda) \in I \right\}, \quad
\lambda^{\sharp}=p\circ \pi^{-1}(\lambda)=\left\{ x \in X \left|\right. (x , \lambda) \in I \right\},
\end{equation}
where $p$ and $\pi$ are the restrictions to $I$ of the natural projections of $X\times \Lambda$.
It is not difficult to prove that  $x^\flat$ and $\lambda^\sharp$ are closed subsets, and that the sets
\begin{equation}
K_{H}=\left\{k \in K \left| \right. kH \cup k^{-1}H \subset HK\right\},
\qquad
H_{K}=\left\{h \in H \left| \right. hK \cup h^{-1}K \subset KH\right\}
\end{equation}
are subgroups of $K$ and $H$, respectively.
Moreover, Helgason \cite{Helgason} proved the following
\begin{theorem}\label{inters.}
The following propositions are equivalent:
\begin{enumerate}
  \item
\begin{equation}
L=K_{H}=H_{K}
\label{eq:transversalitycond}
\end{equation}
  \item The mappings
\begin{equation}
x \mapsto x^\flat, \qquad \lambda \mapsto \lambda^{\sharp}
\end{equation}
are injective.
\end{enumerate}
\end{theorem}
Therefore, under the transversality assumption (\ref{eq:transversalitycond}), any element $x$ of $X$ can be uniquely identified with the subset $x^\flat$ of $\Lambda$, and any element $\lambda\in\Lambda$ can be uniquely identified with the subset $\lambda^\sharp\subset X$. Thus, $X$ and $\Lambda$ are homogeneous spaces  in duality. Exactly as in the case of hyperplanes in $\RM^n$, we have the identification $I=G/L$, with $L=K\cap H$, that is easily proved via the bijection $G/L\to I : g L \mapsto (gK,gH)$. As a consequence, also in the general case we have the following double fibration, which is the generalization of (\ref{eq:doublefibration1}):
\begin{equation}
\begin{diagram}
\node[2]{I =G/(K\cap H)} \arrow{se,t}{\pi} \arrow{sw,t}{p}
\\
\node{X=G/K}  \node[2]{\Lambda=G/H}
\end{diagram}
\label{eq:doublefibration}
\end{equation}
Let $x_0$ and $\lambda_0$ be the origins of $X$ and $\Lambda$, respectively. Let us assume that the sets $x_0^\flat=K/L$ and $\lambda_0^\sharp=H/L$ have positive measures invariant under $K$ and $H$, respectively. Under the transversality assumption (\ref{eq:transversalitycond}), by using the action of $G$ one can uniquely induce a $K$-invariant  ($H$-invariant) positive measure $\mu$ ($m$) on each $x^\flat$ ($\lambda^\sharp$).
Therefore, one can define a Radon transform and its dual as in Eq.~(\ref{eq:Radon+dual}) for the double fibration~(\ref{eq:doublefibration}).

A final remark is in order. The above construction in terms of homogeneous spaces in duality is quite suggestive and represents a good framework for encompassing various examples of tomographic mappings. However, as noted by Helgason himself \cite{Helgason}, it is doubtful that it could capture the large variety of results in this field. In particular, there remain several open problems that should be addressed case by case. In particular, one should find suitable  function spaces on $X$ and $\Lambda$ where the Radon transform $f\mapsto f^\sharp$ is invertible, and -- an even more difficult task -- find an inversion formula.

\section{Examples}

In this section we consider two examples of tomographics mappings a functions$f$ on $\RM^n$, in terms of its integral over certain submanifolds of $\RM^n$ and prove the related inversion formulae.  In the first example we consider the planes of codimension $d$, $0<d<n$,  in $\RM^n$, while in the second one we consider the unit spheres in $\RM^3$.

\subsection{(n-d)-planes in $\RM^n$}
Let us consider the space $\PM^n_d$ of all the affine subspaces $\lambda $ of $\RM^n$ with dimension $n-d$, $0 < d < n$. In this example the homogeneous spaces in duality are $\RM^n$ and $\PM^n_d$ and the double fibration in Eq. (\ref{eq:doublefibration}) reads
\begin{equation}
\begin{diagram}
\node[2]{I= E(n)/(O(d) \times O(n-d))} \arrow{se,t}{\pi} \arrow{sw,t}{p}
\\
\node{\RM^n = E(n)/O(n)}  \node[2]{\!\!\!\!\!\!\!\!\!\!\!\!\!\!\!\!\!\PM^n_d = E(n)/(O(d) \times E(n-d))}
\end{diagram}
\label{eq:doublefibrationspheres}
\end{equation}
Indeed,  $G=E(n)$ and $K=O(n)$, the isotropy group of the origin $0$, hence $X=\RM^n=E(n)/O(n)$. Moreover, the isotropy group of  a reference subspace $\lambda_0 \in \PM^n_d$ passing trough the origin is $H= O(d) \times E(n-d)$,
and thus $\Lambda=\PM^n_d=E(n)/ (O(d) \times E(n-d))$. Finally, $L=K\cap H=
O(d) \times O(n-d)$, so that $I=E(n)/(O(d) \times O(n-d))$.

The Radon transform has the same form~(\ref{sharp lambda1}) of the codimension-one case,
where $\d m(x)$ is the Euclidean measure on the codimension-$d$ hyperplane $\lambda$. First of all notice that, since a codimension-one hyperplane can be viewed as a disjoint union of parallel codimension-$d$ planes, it is obvious that the transform $f\mapsto f^\sharp$ is injective. In other words, the information recorded is redundant and we can recover the function $f$ by integrating away $d-1$ degrees of freedom and then inverting the usual Radon transform on hyperplanes. We will follow here a more direct, different approach.

Observe that $\lambda$ can be obtained by the intersection of $d$ orthogonal hyperplanes in $\RM^n$, namely $\lambda$ can be described as the space of the solutions of the following linear system
\begin{equation}
\label{system}
\lambda =\{ x\in\RM^n \,|\, X_j=\xi_j \cdot x, \; j=1, \dots, d\}.
\end{equation}
where $X_j \in \RM$ and $\xi_j \in \SM^{n-1}$, $\xi_j \cdot \xi_k=\delta_{j,k}$,   with $j,k =1,\dots, d$. Of course, this representation of $\lambda$ is not unique, in fact we can obtain $\lambda$ as the solution of another linear system
\begin{equation}\label{system2}
\lambda =\{ x\in\RM^n \,|\, Y_j=\eta_j \cdot x, \; j=1, \dots, d\}.
\end{equation}
where for all $i=1,\dots, d$
\begin{equation}\label{gauge}
\eta_{i}=\sum_{j=1}^d a_{i j }\xi_{j},
\qquad Y_{i}=\sum_{j=1}^d a_{i j}X_{j},
\end{equation}
and $A=(a_{i j})
\in GL(d,\RM)$.
Let us compute
\begin{equation}
\eta_{i}\cdot \eta_j=\sum_{k,l=1}^d a_{i k }a_{j l }\xi_{k}\cdot \xi_{l}=(A A^T)_{i j}.
\end{equation}
Note that, since  $\eta_i \in \SM^{n-1}$, one must have that  $(A A^T)_{i i}=1$ for all  $i=1,\dots, d$.

Let $f \in \mathcal{S}(\RM^n)$ and define the tomographic transform of $f$ on $\lambda \in \PM^n_d$ as follows (with the usual abuse of notation)
\begin{equation}\label{radon n-d}
f^\sharp (\lambda)=f^\sharp
(Y,\eta) =\int_{\RM^n}f(x)\, \prod_{j=1}^{d}\delta(Y_j-\eta_j \cdot x)\, \sqrt{\det
(\eta_k \cdot \eta_l)
}\; \d x,
\end{equation}
where $Y=(Y_1, \dots, Y_d)\in \RM^d$ and $\eta=(\eta_1, \dots, \eta_d)\in(\SM^{n-1})^d$.
The definition of the Radon transform in (\ref{radon n-d}) does not depend on the parametrization of $\lambda$. Indeed,
by writing the parameterization $(Y,\eta)$ in terms of a parameterization with orthogonal hyperplanes $(X,\xi)$ through (\ref{gauge}), we get
\begin{eqnarray}
& & \prod_{j=1}^{d}\delta(Y_j-\eta_j \cdot x)\, \sqrt{\det (\eta_k \cdot \eta_l)}  =
\int_{\RM^d} \e^{\i \sum_j t_j (Y_j-\eta_j\cdot x)} \left| \det A\right| \; \frac{\d t}{(2\pi)^d}
\nonumber \\
&  & \quad  =
\int_{\RM^d} \e^{\i \sum_j t_j \sum_l a_{jl}(X_l-\xi_l\cdot x)} \left| \det A\right| \; \frac{\d t}{(2\pi)^d}
= \int_{\RM^d} \e^{\i \sum_l \left( \sum_j a_{jl} t_j \right) (X_l-\xi_l\cdot x)} \left| \det A\right| \; \frac{\d t}{(2\pi)^d}
\nonumber \\
&  & \quad  =
\int_{\RM^d} \e^{\i \sum_l s_l (X_l-\xi_l\cdot x)}  \; \frac{\d s}{(2\pi)^d}
= \prod_{l=1}^{d}\delta(X_l-\xi_l \cdot x) .
\end{eqnarray}
where we set $s=A^T t$. Therefore,
\begin{equation}
f^\sharp (Y,\eta) =f^\sharp (X,\xi)= \int_{\RM^n}f(x)\, \prod_{j=1}^{d}\delta(X_j-\xi_j \cdot x) \; \d x
=\int_{\lambda} f(x)\; \d m(x) = f^{\sharp}(\lambda),
\end{equation}
as given in (\ref {sharp lambda1}).

Let us consider also the dual of the Radon transform. If $g \in \mathcal{S}(\PM^n_d)$, we define
\begin{equation}
g^\flat(x)=\int_{(\RM \times \SM^{n-1})^d}g(Y,\eta)\, \prod_{j=1}^{d}\delta(Y_j-\eta_j \cdot x) \;\frac{\d Y\, \d\eta} {\sqrt{\det
(\eta_k \cdot \eta_l )
}} .
\end{equation}
\begin{theorem}
Let $f \in \mathcal{S}(\RM^n)$. Then
\begin{equation}\label{inversion formula d}
f(x)=c_{n,d}(-\Delta)^{\frac{n-d}{2}}(f^\sharp)^\flat(x),
\qquad \text{where} \quad
c_{n,d}=\frac{\Gamma\left(\frac{n-1}{2}\right)^d \Gamma\left(\frac{d}{2}\right)}{2^n \pi^{(dn-d+n)/2}\Gamma\left(\frac{n-d}{2}\right)},
\end{equation}
with $\Gamma$ denoting the Euler gamma function.
\end{theorem}
\begin{proof}
First of all observe that if $n=1$ then necessarily $d=1$, and the statement (\ref{inversion formula d}) is trivial. So assume that $n \geq1$.
Let $0< \alpha < n$, then if $V(x)=1/|x|$, with $x\in\RM^n_*$, one can easily prove that \cite{Lieb}
\begin{equation}
\widehat{V^{n-\alpha}}=b_{n,\alpha}V^{\alpha},
\qquad
\text{where}
\quad b_{n,\alpha}=2^{\alpha}\pi^{n/2}\frac{\Gamma\left(\frac{\alpha}{2}\right)}{\Gamma\left(\frac{n-\alpha}{2}\right)}.
\end{equation}
Let $p>0$ and let us compute the following quantity
\begin{equation}
(-\Delta)^p V^d(x)= \int_{\RM^n}\frac{\d k}{(2\pi)^n}\; |k|^{2p}\, \widehat{V^d}(k)\, \e^{\i k \cdot x}
\end{equation}
so if $2p=n-d$, we obtain
\begin{equation}
(-\Delta)^{\frac{n-d}{2}} V^d(x)=b_{n,n-d}\,\delta(x).
\end{equation}
Now let us compose the Radon transform with its dual. As in Lemma~\ref{sharp flat} we obtain
\begin{equation}
(f^\sharp)^\flat(x)=a_n^d\, (f \ast V^d)(x),
\qquad
a_n=\int_{\RM}\frac{\d s}{2\pi}\int_{\SM^{n-1}} \d\xi\; \e^{\i s\xi \cdot \omega}=\frac{2 b_{n,1}}{2\pi}=\frac{2 \pi^{\frac{n-1}{2}}}{\Gamma\left(\frac{n-1}{2}\right)},
\end{equation}
and again $a_n$ is independent of $\omega \in \SM^{n-1}$.
Therefore, we have
\begin{eqnarray}
(-\Delta)^{\frac{n-d}{2}}(f^\sharp)^\flat
& = & a_n^d\, (-\Delta)^{\frac{n-d}{2}}( f \ast V^d)
= a_n^d \,f \ast \left((-\Delta)^{\frac{n-d}{2}}V^d\right)
\nonumber\\
& = &
a_n^d\, b_{n,n-d}\,f \ast \delta
= a_n^d\, b_{n,n-d}\, f
\end{eqnarray}
and so we have Eq.~(\ref {inversion formula d}) with
$c_{n,d}=1/(a_n^d\, b_{n,n-d})$,
and this concludes the proof.
\end{proof}

\subsection{Spherical means in $\RM^3$}
Let us consider the case, studied by John \cite{John}, of a tomographic  transform that maps a function $f$ on $\RM^n$ into a function on the spheres in $\RM^n$ with a fixed radius $r$. For simplicity we consider the case $n=3$ and $r=1$.

The homogeneous spaces in duality, $\RM^3$ and $\tilde{\RM}^3 =$ set of all unit spheres in $\RM^3$ (parametrized by their centers), are described  by the following double fibration:
\begin{equation}
\begin{diagram}
\node[2]{I= E(3)/O(2)} \arrow{se,t}{\pi} \arrow{sw,t}{p}
\\
\node{\RM^3 = E(3)/O(3)}  \node[2]{\tilde{\RM}^3 = E(3)/\tilde{O}(3)}
\end{diagram}
\label{eq:doublefibrationdplanes}
\end{equation}
Indeed,  $G=E(3)$ and $K=O(3)$, the isotropy group of the origin $0$, hence $X=\RM^3=E(3)/O(3)$. Moreover, the isotropy group of  a reference unit sphere $\lambda_0 \in \tilde{\RM}^3$ passing trough the origin is $H= \tilde{O}(3)$ (rotation around the center of $\lambda_0$),
and thus $\Lambda=\tilde{\RM}^3=E(3)/ \tilde{O}(3)$. Finally, $L=K\cap H=
O(3)\cap \tilde{O}(3)=O(2)$ (rotation around the line joining $0$ and the center of $\lambda_0$), so that $I=E(3)/O(2)$.

\begin{figure}[t]
\includegraphics[width=0.45\textwidth]{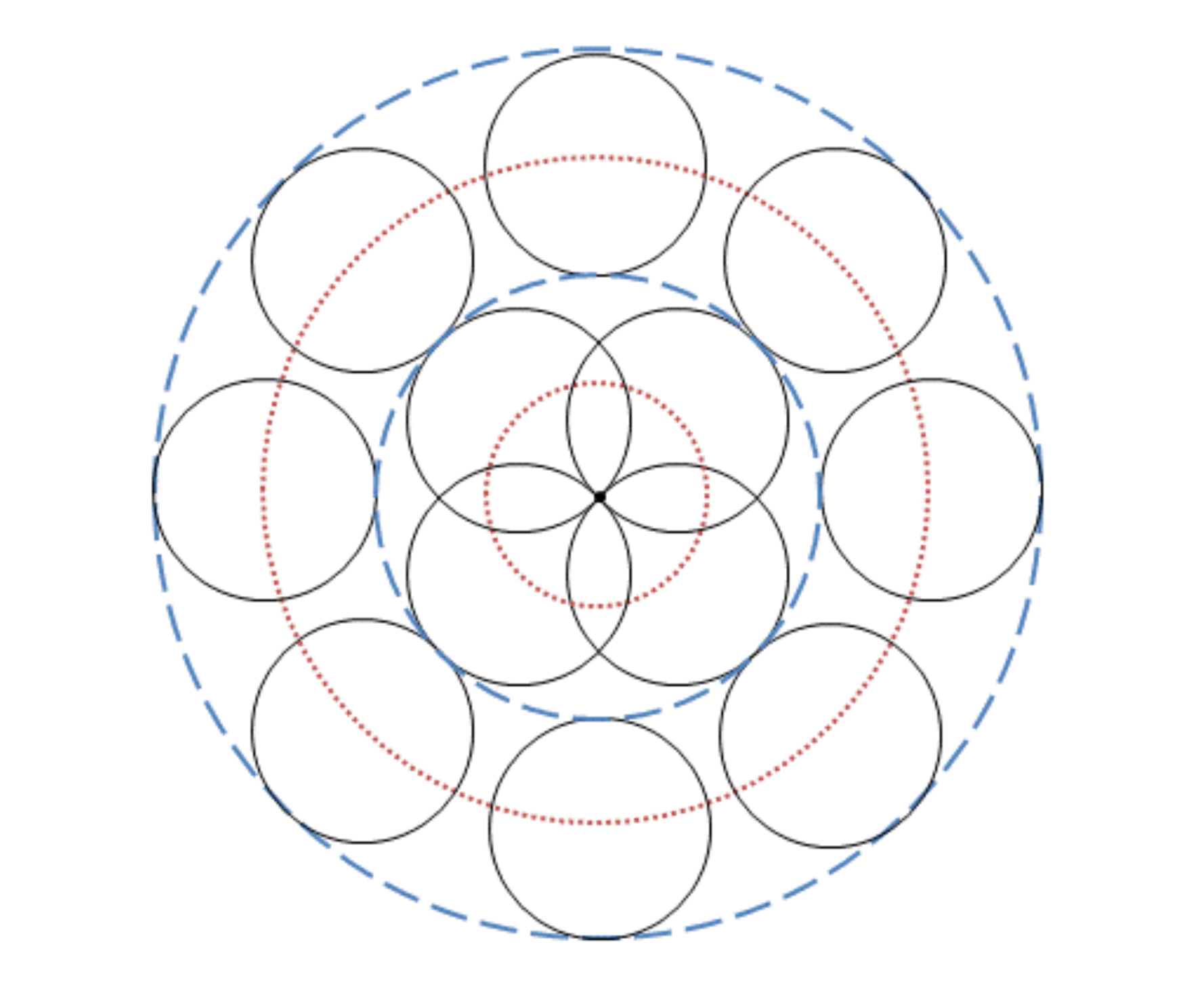}
\caption{Inversion formula (\ref{inv form sph}). Dotted lines are spheres with an odd radius, while dashed lines are spheres with an even radius.}
\label{fig:2}
\end{figure}

We define the tomogram of $f$ as its average over the unit sphere with center $x$
\begin{equation}
f^\sharp(x)=\frac{1}{\left|\SM^{2}\right|}\int_{\SM^2}f(x+\xi)\; \d \xi ,
\end{equation}
where
$\left|\SM^{2}\right|=4\pi$
is the surface of $\SM^{2}$.
We look for an inversion formula, namely we want to reconstruct the function $f$ using its tomograms on all possible unit spheres in $\RM^3$. This problem is equivalent to the following integral equation for the function $f$:
\begin{equation}\label{int eqn}
\frac{1}{\left|\SM^{2}\right|}\int_{\SM^2}f(x+\xi)\; \d \xi=g(x).
\end{equation}
In Ref.~\cite{John} John proved  that the solution of the integral equation (\ref{int eqn}) is unique, under mild regularity assumptions on $f$ and $g$. Therefore, also this tomographic mapping is one to one on $\mathcal{S}(\RM^3)$. We will now  explicitly construct its inverse.
To this purpose, consider the  following Cauchy problem for the wave equation
\begin{equation}\label{wave 1}
\partial^2_{t}v-\Delta v  =  0, \qquad
v(x,0)  =  0, \qquad
\partial_{t}v(x,0)  =  g(x) .
\end{equation}
The solution of this problem is
\begin{equation}
v(x,t)=tI^{(g)}(x,t),
\qquad
I^{(g)}(x,t)=\frac{1}{\left|\SM^{2}\right|}\int_{\SM^2}g(x+t \xi)\; \d \xi.
\end{equation}
We construct the function $u(x,t)$ as follows
\begin{equation}
u(x,t)=-\sum_{k=0}^{+\infty}\left[\partial_{t}v(x,2k+1+t)-\partial_{t}v(x,2k+1-t)\right]
\end{equation}
and define $f(x):=\partial_{t} u(x,0)$. Observe that $u(x,0)  =  0$, thus $u$ solves the Cauchy problem
\begin{equation}\label{wave 2}
\partial^2_{t}u-\Delta u  =  0 , \qquad
u(x,0)  =  0, \qquad
\partial_{t}u(x,0)  =  f(x) ,
\end{equation}
whence $u(x,t)=tI^{(f)}(x,t)$. In particular, we have
\begin{equation}
u(x,1) =I^{(f)}(x,1)
=-\sum_{k=0}^{+\infty}\left[\partial_{t}v(x,2k+2)-\partial_{t}v(x,2k)\right]
=\partial_{t}v(x,0)=g(x) \qquad
\end{equation}
and so we have proved that the function $f$ solves Eq.~(\ref{int eqn}). Moreover,
\begin{eqnarray}
f(x) &=& \partial_{t}u(x,0)
= -2\sum_{k=0}^{+\infty} \partial_{t}^2 v(x,2k+1)
\nonumber\\ &=&
-2\sum_{k=0}^{+\infty} \Delta v(x,2k+1)
= -2 \Delta \sum_{k=0}^{+\infty} (2k+1)I^{(g)}(x,2k+1) \nonumber\\
&=&  -2 \Delta \sum_{k=0}^{+\infty}\frac{2k+1}{\left|\SM^{2}\right|^2}\int_{\SM^2}\d \xi \int_{\SM^2}\d \eta\; f(x+(2k+1)\xi+\eta).
\end{eqnarray}
Therefore, we finally obtain
\begin{equation}
f(x)=  -2 \Delta \sum_{k=0}^{+\infty}\frac{2k+1}{\left|\SM^2\right|}\int_{\SM^2}\d \xi f^\sharp(x+(2k+1)\xi),
\label{inv form sph}
\end{equation}
which is a sum of the spherical means over all the spheres centered at $x$ and with odd radius.
See Fig.~\ref{fig:2}.

\section{Tomograms on hypersurfaces}
\label{sec-submanifolds}

A simple  mechanism that allows nonlinear generalizations of the
Radon transform is the combination of the standard transform with a
diffeomorphism of the underlying $\mathbb{R}^n$ space~\cite{tomocurved}. Let us
consider a function $f(q)$ on the $n$-dimensional space $q
\in \mathbb{R}^n$. The problem is to reconstruct $f$ from
 its integrals over an $n$-parameter family of submanifolds of codimension one.

We can construct such a family by diffeomorphic deformations of
the hyperplanes (in the $x\in\mathbb{R}^n$ space)
\begin{equation}
X -\mu \cdot x = 0, \label{eq:planes}
\end{equation}
with $X\in\mathbb{R}$ and $\mu\in\mathbb{R}^n$. Let us
consider a diffeomorphism of $\mathbb{R}^n$
\begin{equation}
q\in\mathbb{R}^n\mapsto x=\varphi(q)\in\mathbb{R}^n.
\end{equation}
The hyperplanes (\ref{eq:planes}) are deformed by $\varphi$ into a
family of submanifolds (in the $q$ space)
\begin{equation}
X -\mu \cdot \varphi(q) = 0.
\label{eq:submanifold}
\end{equation}
The case $n=2$ is displayed in Fig.\ \ref{fig:diffeo}.

\begin{figure}[t]
\includegraphics[width=0.6\textwidth]{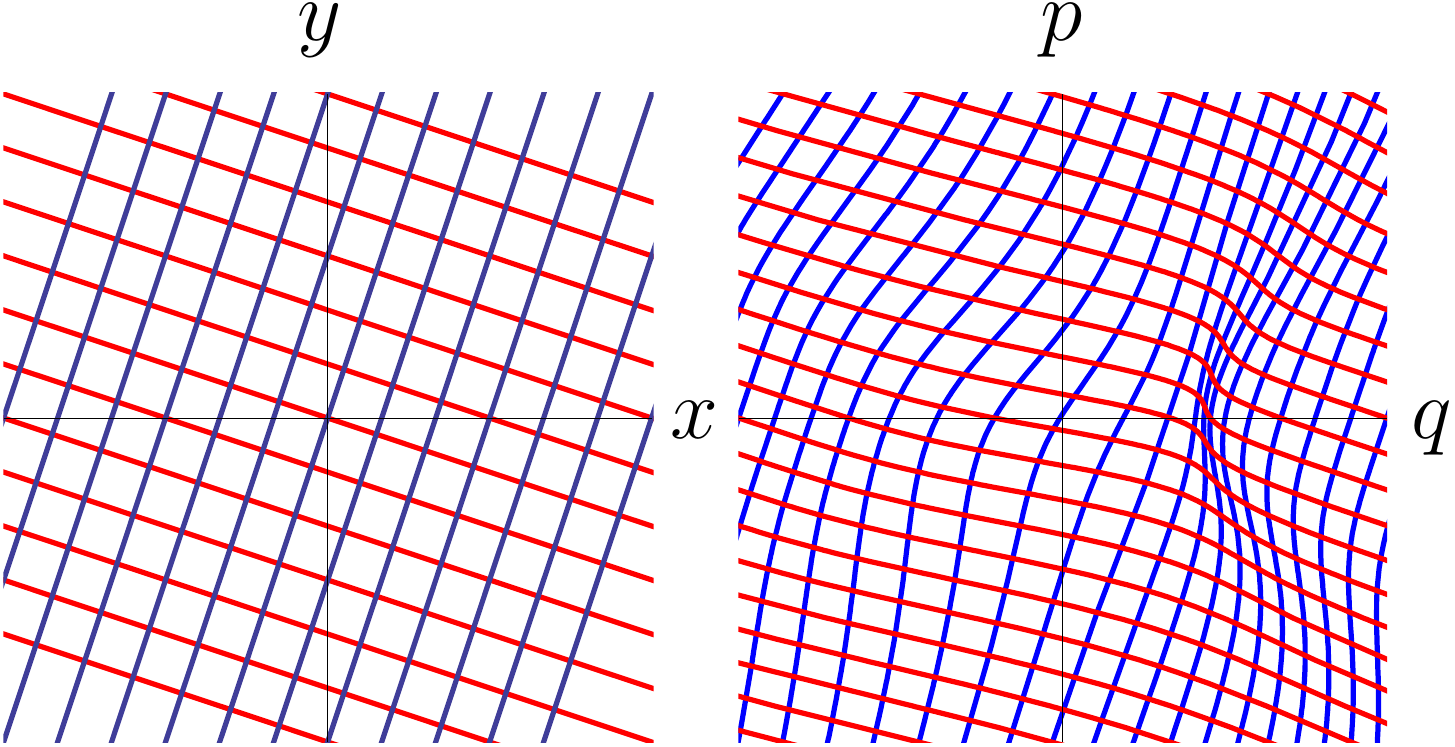}
\caption{Diffeomorphism of the plane: $(q,p)\in
\mathbb{R}^2 \to (x,y)=\varphi(q,p) \in \mathbb{R}^2$.}
\label{fig:diffeo}
\end{figure}

Given a probability density $\tilde f (x)$ on the $x$ space, the
$M^2$-transform can be rewritten as
\begin{equation}
\tilde{f}^{M^2}(X,\mu) = \int_{\mathbb{R}^n} \tilde f(x)\,  \delta(X-\mu\cdot x) \;  \d x
= \int_{\mathbb{R}^n} \tilde f(\varphi(q))\, \delta(X-\mu\cdot
\varphi(q))\,J(q)\; \d q ,
\label{eq:radondefns0}
\end{equation}
where
\begin{equation}\label{eq:Jacobian}
J(q)=\left|\frac{\partial x_i}{\partial
q_j}\right|=\left|\frac{\partial \varphi_i(q)}{\partial
q_j}\right|
\end{equation}
is the Jacobian of the transformation.

Observe now that
$ \tilde f(x)\; \d  x = \tilde f(\varphi(q))\, J(q)\; \d q$,
whence
\begin{equation}
f (q) = \tilde f(\varphi(q))\, J(q)
\end{equation}
is a probability density. Therefore the tomograms along the codimension-1  submanifolds~(\ref{eq:submanifold}) are given by
\begin{equation}
f^\varphi (X,\mu) = \left\langle \delta(X-\mu \cdot \varphi(q))\right\rangle_f
= \int_{\mathbb{R}^n} f(q)\, \delta(X-\mu\cdot \varphi(q))\; \d q ,
\label{eq:radondefns}
\end{equation}
with $X\in \mathbb{R}$ and $\mu \in \mathbb{R}^n_*$.

The inverse transform follows by (\ref{M^2 inversion formula}):
\begin{equation}
f(q) = \tilde f(\varphi(q))\, J(q)
=  \int_{\mathbb{R}^{n+1}} f^\varphi (X,\mu) \, J(q)\, \e^{\i(X-\mu\cdot
\varphi(q))}\; \frac{\d X\, \d\mu}{(2\pi)^n} , \quad
\label{eq:invgens}
\end{equation}
with a modified kernel
\begin{equation}\label{eq:kernel}
K_\varphi(q;X,\mu) = J(q)\, \e^{\i(X-\mu\cdot \varphi(q))}=
\left|\frac{\partial \varphi_i(q)}{\partial q_j}\right|\,
\e^{\i(X-\mu\cdot \varphi(q))}.
\end{equation}
Therefore, a probability density distribution on $\mathbb{R}^n$
\begin{eqnarray}
f(q) \geq 0, \quad \int_{\mathbb{R}^n} f(q) \; \d q =1,
\label{eq:denscond2}
\end{eqnarray}
produces tomograms $f^\varphi (X,\mu)$ that are probability
densities
\begin{eqnarray}
f^\varphi (X,\mu)  \geq 0, \quad \int_{\mathbb{R}} f^\varphi (X,\mu) \; \d X
=1, \quad \forall \mu\in\mathbb{R}^n_*.
\label{eq:normcond2}
\end{eqnarray}
We will now consider two applications of these
generalizations of the Radon transform. For additional examples see \cite {tomogram,tomocurved}.

\subsection{Circles in the plane}
\label{sec-circles}

In the punctured $(x,y)$ plane without the origin $(0,0)$,  the conformal inversion
\begin{equation}\label{eq:conformal}
(x,y)=\varphi(q,p)=\left(\frac{q}{q^2+p^2},\frac{p}{q^2+p^2}\right),
\end{equation}
maps the family of  lines
$X -\mu x - \nu y = 0$
into a family of circles
\begin{equation}
X (q^2+p^2) -\mu q  - \nu p = 0,
\label{eq:circles}
\end{equation}
centered at
$C=\left(\mu/ 2X, \nu/2X \right) $
and passing through the origin (see Fig.\ \ref{fig:circles}). When
$X=0$ they degenerate into lines through the origin.
The Jacobian reads
\begin{equation}
J(q,p)= (q^2+p^2)^{-2},
\end{equation}
whence the transformation is a diffeomorphism of the punctured plane
$\RM^2_*$ onto itself. The  singularity of the
transformation at the origin $(0,0)$ is irrelevant for  tomographic
integral transforms because it only affects a zero measure set.

\begin{figure}[t]
\includegraphics[width=0.4\textwidth]{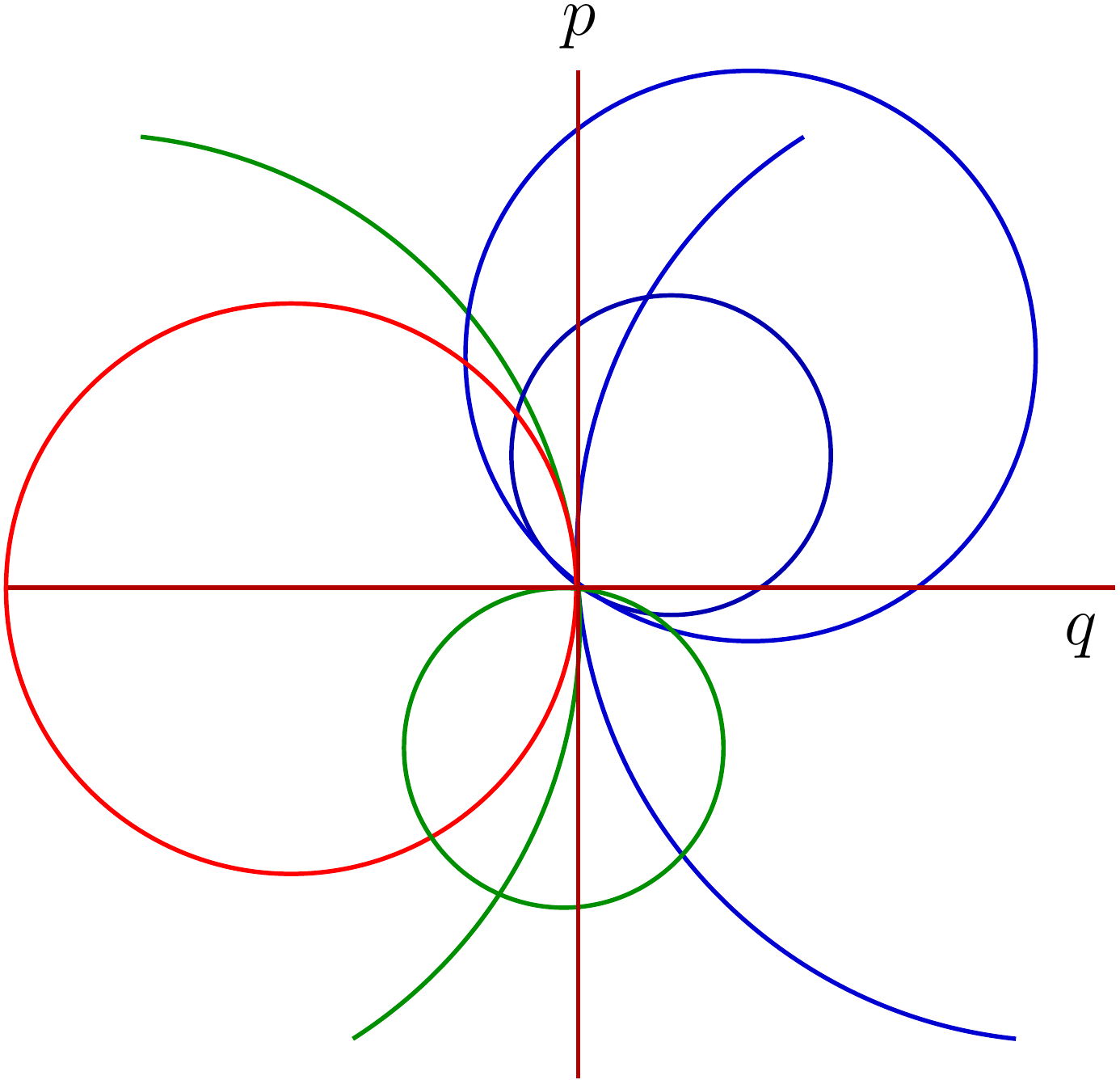} \quad
\includegraphics[width=0.5\textwidth]{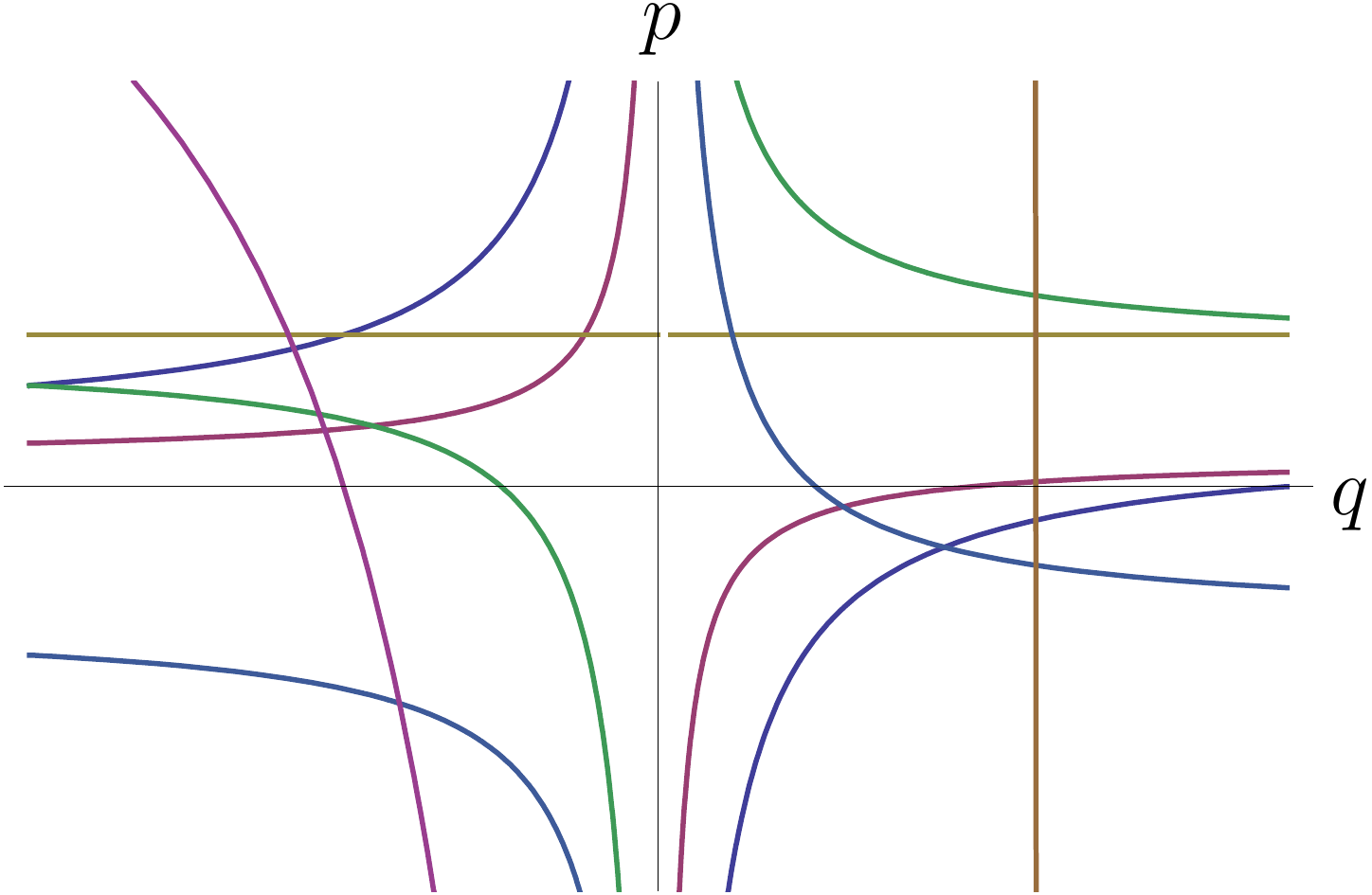}
\caption{Circular tomography (left panel) and hyperbolic tomography (right panel).
All circles pass through the origin and all hyperbolas have the vertical axis as asymptote.}
\label{fig:circles}
\end{figure}

Equations (\ref{eq:radondefns})-(\ref{eq:invgens}) become
\begin{eqnarray}
f^{\varphi}(X,\mu,\nu) = \left\langle
\delta\left(X-\frac{\mu q}{q^2+p^2} -\frac{\nu p}{q^2+p^2}\right)\right\rangle_f ,
\label{eq:radondefnc}
\end{eqnarray}
\begin{equation}
f(q,p) = \int_{\mathbb{R}^3} f^{\varphi}(X,\mu)\,
\frac{\e^{\i\left(X-\frac{\mu q}{q^2+p^2} -\frac{\nu
p}{q^2+p^2}\right)}}{{(2\pi)^2}(q^2+p^2)^2}\,
{\d X \,
\d\mu\, \d\nu} .
\label{eq:invgenc}
\end{equation}

\subsection{Hyperbolas in the plane}
\label{sec-hyperbolas}

In the  $(x,y)$ plane  the family of  lines
$ X -\mu x - \nu y = 0$
is mapped into a family of hyperbolas
\begin{equation}
X -\frac{\mu}{q}  - \nu p = 0,
\label{eq:hyperbola}
\end{equation}
with asymptotes
$q= 0$ and  $p= X/\nu$ ,
by the transformation
\begin{equation}\label{eq:inversion}
(x,y)=\varphi(q,p)=\left(
q^{-1},p\right).
\end{equation}
For $\mu>0$ the hyperbolas are in the second and fourth quadrants,
while for $\mu<0$ they are in the first and third quadrants (see
Fig.\ \ref{fig:circles}). When $\mu=0$ or $\nu=0$ they degenerate
into horizontal or vertical lines, respectively.
The Jacobian reads
\begin{equation}
J(q,p)= q^{-2}
\end{equation}
whence the transformation is a diffeomorphism in the cut plane
without the axis $(0,y)$.

Equations (\ref{eq:radondefns})-(\ref{eq:invgens}) become
\begin{equation}
f^\varphi(X,\mu,\nu) = \left\langle \delta\left(X-\frac{\mu}{q} -\nu p\right)\right\rangle_f ,
\label{eq:radondefnh}
\end{equation}
\begin{equation}
f(q,p) = \int_{\mathbb{R}^3} f^\varphi(X,\mu)\, \frac{1}{q^2}\,
\e^{\i\left(X-\frac{\mu}{q}-\nu p\right)}\; \frac{\d X\, \d\mu\, \d\nu}{(2\pi)^2} .
\label{eq:invgenh}
\end{equation}

\section{Quantum tomography}
We now consider a possible generalization of the above classical techniques to the quantum case.
The most direct path is that of using  the phase-space description of quantum mechanics through Wigner quasi-distribution functions \cite{Wig32,Moyal,Hillary84}. To this purpose, let us consider the phase space $T^*\RM^n\simeq \RM^{2n}$ of a classical system with $n$ degrees of freedom. Let us assume that the (statistical) state of the system is given by a probability density function $f(q,p)$ in $\mathcal{S}(\RM^{2n})$. The expectation value of an observable, which is represented by a function $\sigma(q,p)$ on the phase space, in the state $f$ is given by
\begin{equation}
\langle \sigma \rangle_f = \int_{\RM^{2n}} \sigma(q,p)\, f(q,p) \; \d q\,\d p .
\end{equation}
The $M^2$-transform (\ref{M^2transform}) of the density $f$, and its inverse (\ref{M^2 inversion formula}) read
\begin{eqnarray}
f^{M^2}(X,\mu,\nu) &=& \langle\delta(X-\mu \cdot q-\nu\cdot p) \rangle_f
= \int_{\RM^{2n}}  f(q,p)\, \delta(X-\mu \cdot q-\nu\cdot p)\; \d q\, \d p, \nonumber\\
f(q,p) &=& \int_{\RM^{2n+1}} f^{M^2}(X,\mu,\nu)\, \e^{\i(X-\mu \cdot q- \nu\cdot p)} \; \frac{\d X \d\mu \d\nu}{(2\pi)^{2n}}.
\label{eq:m2ps}
\end{eqnarray}
with $X\in\RM$ and $(\mu, \nu) \in\RM^{2n}_*$. Note that the tomogram  $f^{M^2}_{(\mu,\nu)}$ is nothing but the probability density function of the observable $X(q,p)=\mu \cdot q+\nu\cdot p$.
A very special set of observables are the canonical coordinates: position $q=(q_1, \dots, q_n)$ and momentum $p=(p_1,\dots,p_n)$. One obtains their density directly from the tomograms (\ref{eq:m2ps}) along the Cartesian axes. At $\mu_i= \delta_{i,j}$ and $\nu=0$ one has $X(q,p)=q_j$ and thus
\begin{eqnarray}
& & f^{M^2}(X,(\delta_{i,j}),0)
= \int_{\RM^{2n}}  f(q,p)\, \delta(X-q_j)\; \d q\, \d p
\nonumber\\
& &\quad = \int_{\RM^{2n-1}}  f((q_1,\dots,q_{j-1},X,q_{j+1},\dots,q_n), p)\;
\d q_1 \dots \d q_{j-1} \d q_{j+1}\dots \d q_n \, \d p,\qquad
\end{eqnarray}
that is the marginal density of the position component $q_j$. Analogously, at $\mu=0$ and $\nu_i= \delta_{i,j}$ one has $X(q,p)=p_j$ and
\begin{equation}
f^{M^2}(X,0,(\delta_{i,j}))
= \int_{\RM^{2n}}  f(q,p)\, \delta(X-p_j)\; \d q\, \d p
\end{equation}
is the marginal density of $p_j$.

In quantum mechanics (more precisely, in its Schr\"{o}dinger representation)
the Cartesian canonical coordinates $q$ and $p$ are promoted to operators, $\hat{q}$ and $\hat{p}$, that act in the Hilbert space of square summable functions $L^2(\RM^n)$ endowed with the scalar product
\begin{equation}
\langle \psi, \phi\rangle = \int_{\RM^n} \overline{\psi(x)}\, \phi(x)\; \d x,
\end{equation}
with $\psi,\phi\in L^2(\RM^n)$ and the overline denoting complex conjugation.
Precisely, $\hat{q}$ is the multiplication operator $(\hat{q} \psi)(x)=x \psi(x)$, and the $\hat{p}=-i\hbar \nabla$ is a differential operator, both with maximal domains. Here, $\hbar$ is the Planck constant. Pure quantum states of the system are associated to normalized elements $\psi\in L^2(\RM^n)$, called wave functions, $\|\psi\|_2^2= \langle\psi,\psi\rangle=1$. The expectation value of position (momentum) in the wave function $\psi$ is given by the sandwich $\langle \psi, \hat{q} \psi\rangle$ ($\langle \psi, \hat{p} \psi\rangle$). It  easily follows that $x\mapsto |\psi(x)|^2$ is the marginal probability density  of position (correctly normalized, since $\|\psi\|_2=1$). Moreover, by Fourier transform it is also not difficult to show
that the marginal of momentum is given by $p\mapsto |\hat{\psi}(p/\hbar)|^2/(2\pi\hbar)^n$ (prove it!).

Now we want to define the quantization of any classical observable $\sigma$ (also called symbol), that for simplicity we assume in $\mathcal{S}(\RM^{2n})$. Observe that
\begin{equation}
\sigma(q,p)= \int_{\RM^{2n}}\hat{\sigma}(\eta,y)\, \e^{\i (\eta\cdot q+ y\cdot p) }\;  \frac{\d\eta\,\d y}{(2\pi)^{2n}},
\end{equation}
so, by following Weyl, we can construct an operator out of $\sigma$ by substituting $q$ with the position operator $\hat{q}$ and $p$ with the momentum operator $\hat{p}$. Thus, we are led to consider the unitary operators
\begin{equation}
(y,\eta)\mapsto \e^{\i(\eta \cdot \hat{q}+y\cdot \hat{p})} , \qquad \left(\e^{\i(\eta \cdot \hat{q}+y\cdot \hat{p})}\psi\right)(x)=\e^{\i\eta\cdot (x+\hbar y/2)}\psi(x+\hbar y)
\label{eq:SchrHeis}
\end{equation}
(the Schr\"{o}dinger representation on $L^2(\mathbb{R}^n)$ of the Heisenberg group \cite{Folland}), with $(y,\eta)\in\RM^{2n}$ and $\psi\in L^2(\mathbb{R}^n)$.
In conclusion, for any $\sigma\in\mathcal{S}(\RM^{2n})$
we define the \emph{Weyl quantization} of the symbol $\sigma$, the operator
\begin{equation}\label{weyl}
\mathrm{Op}(\sigma)= \int_{\RM^{2n}}\hat{\sigma}(\eta,y)\,\e^{\i(\eta \cdot \hat{q}+y\cdot \hat{p})} \; \frac{\d\eta \d y}{(2\pi)^{2n}},
\end{equation}
whose action on any $\psi\in L^2(\RM^n)$ is
\begin{equation}
(\mathrm{Op}(\sigma)\psi)(x) =  \int_{\RM^{n}} K_{\sigma}(x,y)\, \psi(y) \; \d y
\end{equation}
 with a kernel
\begin{equation}\label{kernel}
K_{\sigma}(x,y)= \int_{\RM^{n}}\sigma\left(\frac{x+y}{2}, \eta \right)\, \e^{\i \eta\cdot (x-y)/\hbar} \;  \frac{\d\eta}{(2\pi\hbar)^{n}} .
\end{equation}
Observe that the kernel $K_{\sigma}$ defined in (\ref{kernel}) is obtained from $\sigma$ by a partial Fourier transform followed by an invertible and measure-preserving change of variables. These operations make sense when $\sigma$ is an arbitrary temperated distribution and define $K_{\sigma}$ as a temperated distribution. Therefore, it is possible to define the quantization of a generic temperated distribution. Later we will use the quantization of the Dirac $\delta$ distribution.

Now we want the inverse mapping of the Weyl quantization, namely we want to obtain the symbol starting from the operator. From~ (\ref{kernel}) it is easy to check that the \emph{dequantization} of the operator $S$ whose kernel is $K(x,y)$ is given by the symbol
\begin{equation}\label{dequantization}
\sigma(q,p)= \mathrm{Op}^{-1}(S) (q,p) = (2\pi\hbar)^n\, W(K)(q,p),
\end{equation}
where
\begin{equation}
W(F)(q,p)= \int_{\RM^{n}}F\left(q+ \frac{y}{2},q- \frac{y}{2}\right)\, \e^{-\i y\cdot p /\hbar}\; \frac{\d y}{(2\pi \hbar )^{n}}
\label{eq:WignerR2n}
\end{equation}
is known as the \emph{Wigner transform} of $F \in L^2(\RM^{2n})$. This is the fundamental connection between Weyl quantization and Wigner transform.
By a straightforward direct computation, that we leave as an exercise, one has the following remarkable relation:
\begin{equation}\label{eq:trace relation}
\Tr(\mathrm{Op}\, (\sigma)\mathrm{Op}(\tau)) :=
\int_{\RM^{2n}} K_{\sigma}(x,y)\, K_{\tau}(y,x)\; \d x\, \d y
= \int_{\RM^{2n}}\sigma(q,p)\, \tau(q, p)\; \frac{\d q\, \d p}{(2\pi \hbar)^n},
\end{equation}
for any $\sigma, \tau \in\mathcal{S}(\RM^{2n})$.
This property  will be crucial for our tomographic framework.

The statistical states of a quantum system are described by positive (self-adjoint) operators with unit trace on the Hilbert space $L^2(\RM^n)$, namely $\rho \geq 0$, $\rho^{*}=\rho$ and $\Tr\rho =1$, that are called \emph{density operators}.  It is easy to check that  a density operator $\rho$ is an integral operator on $L^2(\RM^n)$, and if
\begin{equation}
\rho=\sum_{k} \lambda_k   \psi_k \langle \psi_k, \cdot \rangle
\end{equation}
is the spectral resolution of $\rho$, with $\{\psi_k\}$ an orthonormal eigenbasis and $1\geq \lambda_k \geq \lambda_{k+1}\geq 0$ the corresponding eigenvalues, the trace condition is equivalent to $\sum_k \lambda_k =1$,
and the kernel of $\rho$ is (with a little abuse of notation)
\begin{equation}
\rho(x,y)=\sum_{k} \lambda_k  \psi_k(x) \overline{ \psi_k(y)}.
\end{equation}
The expectation value in the state $\rho$  of an observable associated to a bounded operator $S$ is given by
\begin{equation}
\langle S \rangle_\rho = \Tr(S\, \rho) = \sum_k \lambda_k \langle \psi_k, S\psi_k\rangle.
\label{eq:quantumexp}
\end{equation}
 Therefore, the density matrix $\rho$ is given by a convex combination of wave functions  $\psi_k$ with weights $\lambda_k$.  Physically, it represents a classical statistical mixture of pure quantum states extracted with probabilities $\lambda_k$.
By making use of Eqs.~(\ref{dequantization}) and (\ref{eq:trace relation}), the quantum expectation value (\ref{eq:quantumexp}) of $S= \mathrm{Op}(\sigma)$ reads
\begin{equation}\label{mean value wigner}
\langle \mathrm{Op}(\sigma)\rangle_\rho = \Tr\left(\mathrm{Op}(\sigma)\, \rho \right) = \int_{\RM^{2n}}\sigma(q,p)\, W(\rho)(q, p)\; \d q\, \d p = \langle \sigma \rangle_{W(\rho)}.
\end{equation}
This equation expresses the expectation value of a quantum observable in the quantum state $\rho$ as a classical expectation value of its symbol in the ``classical'' state $W(\rho)$.
In this respect, $W(\rho)$ is called \emph{Wigner quasi-distribution function} of $\rho$ and is a substitute for the nonexistent joint probability density  of momentum and position in the quantum state $\rho$. Indeed, since the uncertainty principle
imposes a limit on the precision with which momentum and position can be determined at the same time in the state $\rho$, it does not make sense to speak of a joint probability distribution for these observables.
In fact, $W(\rho)$ is not a genuine probability density, because it may assume negative values (e.g.\ if $\rho=\psi \langle\psi,\cdot\rangle$, with $\psi \in L^2(\RM^n)$ odd, we have that $W(\rho)(0,0)=-\|\psi\|_2^2 /(\pi \hbar)^n<0$). It is, however, always real and in some sense it tries very hard to be a joint density for momentum and position. The above discussion provides supporting evidence for this heuristic assertion.
Moreover, it is easy to check that $W(\rho)$ has the right marginals:
\begin{eqnarray}
\int_{\RM^n}W(\rho)(q,p)\; \d q &=& \frac{1}{(2\pi \hbar)^n} \hat{\rho} (p/\hbar,p/\hbar)
=\frac{1}{(2\pi \hbar)^n} \sum_k \lambda_k \left|\hat{\psi}_k(p/\hbar)\right|^2 ,
\nonumber \\
\int_{\RM^n}W(\rho)(q,p)\; \d p &=& \rho(q,q)=\sum_k \lambda_k \left|\psi_k(q)\right|^2.
\end{eqnarray}

Now we can apply all the results obtained in Secs.~\ref{sec: Radon} and \ref{sec:M^2}
to reconstruct the state of a quantum system using a tomographic approach.
In the quantum framework the object that we want to reconstruct is the state $\rho$ of the system. Therefore, let us compute the expectation of
\begin{equation}
\delta(X-\mu\cdot\hat{q}-\nu\cdot\hat{p}) := \mathrm{Op}\left(\delta(X-\mu\cdot q-\nu\cdot p)\right)
\label{eq:quadratures}
\end{equation}
in the state $\rho$.
By (\ref{mean value wigner})  and (\ref{eq:m2ps}) one has
\begin{eqnarray}
\left\langle \delta(X-\mu\cdot\hat{q}-\nu\cdot\hat{p}) \right\rangle_{\rho}
= \left\langle \delta(X-\mu\cdot q-\nu\cdot p ) \right\rangle_{W(\rho)} =
W(\rho)^{M^2}(X,\mu,\nu).
\end{eqnarray}
Therefore, by measuring the expectation in the state $\rho$ of the family of operator-valued distributions (\ref{eq:quadratures}),
with $X \in \RM$, $(\mu,\nu)\in\RM_*^{2n}$,  we in fact are gathering  all the tomograms of the Wigner function $W(\rho)$. Therefore, by  applying  the inversion formula (\ref{eq:m2ps}) we can recover $W(\rho)$, that is, up to a normalization factor, the symbol of the density operator $\rho$ [see Eq.~(\ref{dequantization})]. Thus, by quantizing it we can reconstruct the quantum state $\rho$. Explicitly, by using Eqs.~(\ref{dequantization}), (\ref{kernel}) and (\ref{eq:m2ps}), we get
\begin{eqnarray}
\rho(x,y) & = & \int_{\RM^n} W(\rho)\left(\frac{x+y}{2},\eta\right)\, \e^{\i\eta\cdot(x-y)/\hbar} \; \d\eta
\nonumber\\
& = &  \int_{\RM^n} \d\eta\; \e^{\i\eta\cdot(x-y)/\hbar} \int_{\RM^{2n+1}}
\frac{\d X\, \d\mu\, \d\nu}{(2\pi)^{2n}}\; W(\rho)^{M^2} (X,\mu,\nu)
\, \e^{\i (X - \mu \cdot(x+y)/2 -\nu\cdot \eta)}
\nonumber\\
& = &   \int_{\RM^{2n+1}}
W(\rho)^{M^2} (X,\mu,\nu)
\, \e^{\i (X - \mu \cdot(x+y)/2)} \, \delta(\nu- (x-y)/\hbar)\;  \frac{\d X\, \d\mu\, \d\nu}{(2\pi)^{n}}
\nonumber\\
& = &   \int_{\RM^{n+1}}
 W(\rho)^{M^2} (X,\mu,(x-y)/\hbar)
\; \e^{\i (X - \mu \cdot(x+y)/2)}\; \frac{\d X\,  \d\mu}{(2\pi)^{n}}.
\end{eqnarray}
We recall that the tomogram $W(\rho)^{M^2}_{(\mu,\nu)}$ is the marginal probability density of the quantum observable $X_{(\mu,\nu)}(\hat{q},\hat{p})=\mu\cdot\hat{q}+\nu\cdot\hat{p}$. When $\xi=(\mu,\nu)\in \SM^{2n-1}$ the observable $X_\xi$ is nothing but the component of a generalized position, called quadrature, in a rotated frame of the phase space. From the experimental point of view, quadratures of photons can be measured in a typical quantum optics setup by a homodyne detection scheme.
A different scheme was proposed for the longitudinal motion of neutrons \cite{reconstruct06}.
Recently, in a spectacular experiment \cite{konst} quadratures were measured  for reconstructing the transversal motional states of helium atoms, by using a free evolution of the wave packet followed by a position sensitive measurement.

The determination of the quantum state represents a highly
nontrivial problem, whose history can be traced back to the early
days of quantum mechanics, namely to the Pauli problem of reconstructing a wave function from its marginals in position and momentum.
The experimental validation had to wait until quantum
optics opened a new era. In this section we have followed the idea of Vogel and
Risken \cite{Vog-Ris}. Since then many improvements and new
techniques have been proposed \cite{Mancini95,hradil,Marmoopen}. An up-to-date overview can be found in Ref.~\cite{Jardabook}.

\subsection*{Notation}
\vspace{-0.5cm}
\small{\begin{eqnarray*}
f^\sharp, g^\flat & & \text{Radon transform and its dual, see Eq. (\ref{sharp lambda1}) and Eq. (\ref{dual map})} \\
f^\sharp_{\xi} & & \text{Tomogram of $f$ along the direction $\xi$, see Eq. (\ref{tomogram})} \\
f^{M^2} & & \text{$M^2$ transform, see Eq. (\ref{M^2transform})} \\
\mathcal{S}(M) & & \text{Schwartz space of rapidly decreasing functions on $M$,  see Eq.~(\ref{eq:Schwartz})} \\
\delta(x) & & \textrm{Dirac delta distribution at $0$} \\
\hat{f},\check{f} & & \text{Fourier transform and its inverse, see Eqs.~(\ref{fourier transform})-(\ref{inv fourier transform})} \\
W(\rho) && \text{Wigner transform of the state $\rho$, see Eq.~(\ref {eq:WignerR2n}) } \\
\textrm{Op}(\sigma)  && \text{Weyl quantization of the symbol $\sigma$, see Eq. (\ref{weyl})} \\
\PM^{n}_d & & \text{Space of all the $(n-d)$-dimensional affine subspaces of $\RM^n$.\; ($\PM^{n} =\PM^{n}_1$)} \\
\RM^n_* &=& \RM^n\setminus \{0\},\quad \RM^+=[0,+\infty), \quad  \NM_*=\NM \setminus \{0\} \\
\Gamma(z)&=& \int_{0}^{+\infty}\lambda^{z+1}\, \e^{-\lambda}\;\d \lambda, \text{   Euler gamma function}
\end{eqnarray*}}


\begin{theacknowledgments}
We would  like to thank Beppe Marmo, Manolo Asorey, Volodya Man'ko, and Saverio Pascazio for stimulating
discussions.
P.~F.\ would like to thank the organizers,
J.~F.\ Cari\~{n}ena, E.\ Mart\'{i}nez, J.\ Clemente-Gallardo,
J.~N.\ da Costa, and D.\ Mart\'{i}n de Diego, for their kindness in inviting him and for the effort they exerted on the organization of the workshop.
This work is partly supported by the European Union
through the Integrated Project EuroSQIP.
\end{theacknowledgments}



\bibliographystyle{aipproc}   





\end{document}